\documentclass{lmcs} 
\pdfoutput=1

\usepackage{lastpage}
\lmcsdoi{22}{2}{20}
\lmcsheading{}{\pageref{LastPage}}{}{}%
{Jan.~17,~2025}{May~20,~2026}{}

\theoremstyle{plain} 

\usepackage{iftex}
\usepackage{graphicx}
\usepackage[colorlinks=true,allcolors=black,bookmarks=false]{hyperref}%
\usepackage{stmaryrd}
\usepackage{amssymb,amsmath,amsthm}
\usepackage{stmaryrd}
\usepackage{pgf, tikz}
\usepackage{algorithm}
\usepackage{algorithmic}
\usepackage{fancyvrb}
\usepackage{mdframed}
\usepackage{pgfplotstable}
\usepackage{pgfplots}
\usepackage{enumitem}
\usepackage{chngcntr}
\usepackage{capt-of}
\usepackage{cite}
\usepackage{color, colortbl}
\usepackage{bookmark}
\usepackage[utf8]{inputenc}

\theoremstyle{definition}

\newtheorem*{prop*}{Proposition}
\theoremstyle{definition}
\newtheorem{defn}[thm]{Definition}

\definecolor{plot1}{RGB}{70, 116, 193}
\definecolor{plot2}{RGB}{235, 125, 60}
\definecolor{plot3}{RGB}{165, 165, 165}
\definecolor{plot4}{RGB}{252, 190, 45}
\definecolor{plot5}{RGB}{94, 156, 210}
\definecolor{plot6}{RGB}{113, 171, 77}
\definecolor{plot7}{RGB}{156, 72, 25}
\definecolor{plot8}{RGB}{40, 69, 117}
\definecolor{gp2grey}{RGB}{150, 150, 150}

\usetikzlibrary{positioning, quotes, automata, arrows}

{\begin{trivlist}\item[]%
{\textit{Proof sketch.}}}%
{\qed\end{trivlist}}

{\begin{trivlist}\item[]%
{\textit{Proof of #1.}}}%
{\qed\end{trivlist}}

\newcommand{\ttt}{\texttt}

\newcommand{\mtt}{\mathtt}
\newcommand{\mrm}{\mathrm}

\renewcommand{\L}{\mathcal{L}}

\newcommand{\ul}[1]{\underline{#1}}

\newcommand{\tuple}[1]{\langle#1\rangle}

\renewcommand{\bar}[1]{\overline{#1}}

\newcommand{\dder}{\Rightarrow}

\newcommand{\gp}{\texorpdfstring{GP\,2}{}}

\begin{document}


\title[Rule-Based Graph Programs with Imperative Time Complexity]{Rule-Based Graph Programs Matching the Time Complexity of Imperative Algorithms}

\author[Z.~Ismaili Alaoui]{Ziad Ismaili Alaoui}[a,b]
\author[D.~Plump]{Detlef Plump}[c]

\address{University of Liverpool, United Kingdom}	
\email{ziad.ismaili-alaoui@liverpool.ac.uk}  

\address{Most of this author’s work was done while he was affiliated with the University of York.}

\address{University of York, United Kingdom}	
\email{detlef.plump@york.ac.uk}  





\begin{abstract}
We report on recent advances in rule-based graph programming, which allow us to match the time complexity of some fundamental imperative graph algorithms. In general, achieving the time complexity of graph algorithms implemented in conventional languages using a rule-based graph-transformation language is challenging due to the cost of graph matching. Previous work demonstrated that with \emph{rooted} rules, certain algorithms can be implemented in the graph programming language \gp{} such that their runtime matches the time complexity of imperative implementations. However, this required input graphs to have a bounded node degree and (for some algorithms) to be connected. In this paper, we overcome these limitations by enhancing the graph data structure generated by the \gp{} compiler and exploiting the new structure in programs. We present three case studies: the first program checks whether input graphs are connected, the second program checks whether input graphs are acyclic, and the third program solves the single-source shortest-path problem for graphs with integer edge-weights. The first two programs run in linear time on (possibly disconnected) input graphs with arbitrary node degrees. The third program runs in time $\mathrm{O}(nm)$ on arbitrary input graphs, matching the time complexity of imperative implementations of the Bellman-Ford algorithm. For each program, we prove its correctness and time complexity, and provide runtime experiments on various graph classes.
\end{abstract}

\maketitle

\section{Introduction}
\label{s:introduction}
Designing and implementing languages for rule-based graph rewriting, such as GReAT \cite{Agrawal-Karsai-Neema-Shi-Vizhanyo06a}, GROOVE \cite{Ghamarian-deMol-Rensink-Zambon-Zimakova12a}, GrGen.Net \cite{Jakumeit-Buchwald-Kroll10a}, Henshin \cite{Struber-Born-Gill-Groner-Kehrer-Ohrndorf-Tichy17a}, and PORGY \cite{Fernandez-Kirchner-Pinaud19a}, poses significant performance challenges. Typically, programs written in these languages do not achieve the same asymptotic runtime efficiency as those written in conventional imperative languages such as C or Java. The primary obstacle is the cost of graph matching, where matching the left-hand graph $L$ of a rule within a host graph $G$ generally requires time $|G|^{|L|}$, with $|X|$ denoting the size of graph $X$. (Since $L$ is fixed, this is a polynomial.) As a consequence, standard imperative graph algorithms running in linear time (see, for example, \cite{Cormen-Leiserson-Rivest-Stein22a,Skiena20a}) may exhibit non-linear, polynomial runtimes when recast as rule-based graph programs.

To address this issue, the graph programming language \gp{} \cite{Plump12a} supports \emph{rooted} graph transformation rules, initially proposed by D\"orr \cite{dorr1995efficient}. This approach involves designating certain nodes as \emph{roots} and matching them with roots in the host graphs. Consequently, only the neighbourhoods of host graph roots need to be searched for matches, which can often be done in constant time under mild conditions. The \gp{} compiler \cite{Bak15a} maintains a list of pointers to roots in the host graph, facilitating constant-time access to roots if their number remains bounded throughout the program's execution. In \cite{Bak-Plump12a},\, \emph{fast} rules were identified as a class of rooted rules that can be applied in constant time, provided host graphs contain a bounded number of roots and have a bounded node degree.

The first linear-time graph problem implemented by a GP\,2 program was 2-colouring. In \cite{Bak-Plump12a,Bak15a}, it is shown that this program colours connected graphs of bounded node degree in linear time. Since then, the \gp{} compiler has received some major improvements, particularly related to the runtime graph data structure used by the compiled programs \cite{Campbell-Romo-Plump20d}. These improvements made a linear-time worst-case performance possible for a wider class of programs on input graph classes of bounded degree, in some cases even on those of unbounded degree. See \cite{campbell2022fast} for an overview.

Prior to the improvements presented in this paper, besides the issue of input graph classes with unbounded node degree, some programs, such as the aforementioned 2-colouring program, needed a non-linear runtime on disconnected input graphs. In this paper, we present two updates to the \gp{} compiler which allow lifting the constraints that input graphs (1) must be connected and (2) have a bounded node degree. In short, the solution is to improve the graph data structure generated by the compiler. Nodes are now stored in separate linked lists based on their \emph{marks} (red, green, blue, grey or unmarked), and each node comes with a two-dimensional array of linked lists storing all incident edges based on their marks (red, green, blue, dashed or unmarked) and orientation (incoming, outgoing or looping). This enables the matching algorithm to find in constant time a node with a specific mark or an edge with a specific mark and orientation. For instance, if a red node is needed, a single access to the list of red nodes will either locate such a node or confirm its absence. Similarly, if an outgoing green edge is required, a single access to the corresponding linked list will either find such an edge or determine that none exists.

In addition to the new graph data structure, a programming technique is needed to take advantage of the improved storage. In our first case study, we demonstrate how the new approach allows checking in linear time whether arbitrary input graphs are connected. In the second case study, we provide a program that recognises acyclic graphs in linear time. Again, input graphs may be disconnected and have arbitrary node degrees. Finally, our third case study shows how to implement the Bellman-Ford single-source shortest-path algorithm in \gp{} such that it runs in time $\mathrm{O}(nm)$ on arbitrary input graphs (where $n$ and $m$ are the numbers of nodes and edges, respectively), thus matching the complexity of imperative versions of the algorithm. Each case study comes with proofs of the program's correctness and time complexity. In addition, we provide runtime experiments on various graph classes.

Building upon the improvements to the compiler presented in this paper, this work further expands previous studies on fast rule-based graph programs \cite{campbell2022fast,ismaili2024linear,Ismaili-Plump25a} with the following contributions: (1) We present a GP\,2 program checking connectedness in linear time on arbitrary input graphs, eliminating the constraint of \cite{campbell2022fast} that input graphs must have a bounded node degree. We prove that the program is correct and always runs in linear time. In addition, we present empirical runtime measurements on six classes of input graphs. (2) We prove the correctness and time complexity of the linear-time cycle checking program from \cite{Ismaili-Plump25a}. Again, the proofs are complemented by runtime measurements on various classes of input graphs. (3) For the first time, we present a rule-based implementation of Bellman-Ford's single-source shortest-paths algorithm such that the imperative time complexity of $\mathrm{O}(nm)$ is matched. We give full proofs of the program's correctness and time complexity, and also provide runtime experiments confirming the complexity.

The rest of this paper is organised as follows. In Section~\ref{s:preliminaries}, we recall the basic features of the graph programming language GP\,2. In the next section, we briefly describe the GP\,2-to-C compiler and state assumptions on the time complexity of elementary operations of the generated code. In Section ~\ref{s:problem1}, we explain the problem caused by classes of input graphs with an unbounded node degree, while Section~\ref{s:enhancement} introduces the compiler improvements. Section~\ref{s:connectedness}, Section~\ref{s:acyclicity} and Section~\ref{s:bf} present our three case studies. Finally, Section~\ref{s:conclusion} concludes the paper and mentions some topics for future work. In addition, Appendix~\ref{s:appendix} lists the abstract syntax of GP\,2 labels and application conditions.

\section{The Graph Programming Language GP\,2}
\label{s:preliminaries}

This section (largely adapted from \cite{campbell2022fast}) introduces GP\,2, an experimental programming language based on graph-transformation rules. The language was defined in \cite{Plump12a}, an up-to-date version of its concrete syntax can be found in \cite{Bak15a}. 

\subsection{Graphs and Rules}
\label{subsec:graphs_rules}

GP\,2 programs transform input graphs into output graphs, as depicted in Figure \ref{fig:computational-model}. Graphs are directed and may contain parallel edges and loops. Both nodes and edges are labelled.

\begin{figure}[H]
\centering
\begin{tikzpicture}
\node at (6.0,0.75) {GP\,2};
\node at (6.0,0.25) {Program};
\draw [->] (5.0,0.5) -- (7,0.5);
\node at (8,1.7) {Input Graph};
\draw [->] (8,1.5) -- (8,1);
\draw [fill=black!10] (7,0) rectangle node{Execute} (9,1);
\draw [->] (9,0.5) -- (11.0,0.5);
\node at (10.0,0.75) {Output};
\node at (10.0,0.25) {Graph};
\end{tikzpicture}
\caption{GP\,2 computational model.}
\label{fig:computational-model}
\end{figure}
\begin{defn}[Graph]
\label{def:graph}
Let $\L$ be a set of labels. A \emph{graph}\/ over $\L$ is a system $G = \tuple{V_G,E_G,s_G,t_G,l_G,m_G}$, where $V_G$ and $E_G$ are finite sets of nodes (or vertices) and edges, $s_G\colon E_G \to V_G$\/ and $t_G\colon E_G \to V_G$\/ are source and target functions for edges, $l_G\colon V_G \to \L$ is the partial node labelling function, and $m_G\colon E_G \to \L$ is the (total) edge labelling function.
\end{defn}
A graph $G$ is \emph{totally labelled}\/ if $l_G$ is a total function. Unlabelled nodes only occur in the interfaces of rules and are used to relabel nodes (see Definition \ref{def:rule} below). There is no need to relabel edges as they can always be deleted and reinserted with different labels.
    
We distinguish two types of graphs, \emph{host graphs}\/ and \emph{rule graphs}. The former are the graphs on which GP\,2 programs are executed while the latter are the graphs occurring in rules. Figure \ref{fig:host-label-grammar} and Figure \ref{fig:rule-label-grammar} in Appendix \ref{s:appendix} give grammars defining the abstract syntax of labels in host graphs and rule graphs, respectively.

Both nodes and edges are labelled with list expressions denoting integers and character strings. This includes the empty list as a special case; by convention, we draw such items in pictures without labels. 
Labels are of type \texttt{list}, \texttt{atom}, \texttt{int}, \texttt{string} or \texttt{char}, where \texttt{atom} is the union of \texttt{int} and \texttt{string}, and \texttt{list} is the type of a list of atoms. Lists of length one and strings of length one are equated with their entry. As a consequence, every expression of category Atom, Integer, String or Char as defined by the grammar in Figure \ref{fig:rule-label-grammar} can be considered as a list expression. The label subtype hierarchy of GP\,2 is shown in Figure \ref{fig:subtypes}. 

Labels in host graphs as defined by the grammar in Figure \ref{fig:host-label-grammar} form a subset of the labels allowed in rule graphs. They represent constant values in that they do not contain variables or unevaluated expressions.

Besides carrying expressions, nodes and edges can be \emph{marked} red, green or blue. In addition, nodes can be marked grey and edges can be dashed. When we draw graphs, marks are represented graphically. The wildcard mark \texttt{any}, represented by the magenta colour, can only be used in rules and matches arbitrary marks in host graphs. Marks are convenient, among other things, to record visited items during a graph traversal and to encode auxiliary structures in graphs. The programs in the case studies of this paper use marks extensively.

\begin{figure}[!t]
\centering
\begin{tikzpicture}
\node (li) at (0,0) {\ttt{list}};
\node (at) at (0,-1) {\ttt{atom}};
\node (in) at (-1,-2) {\ttt{int}};
\node (st) at (1,-2) {\ttt{string}};
\node (ch) at (1,-3) {\ttt{char}};

\path (at) -- (li) node[pos=0.5,sloped] {$\subseteq$};
\path (in) -- (at) node[pos=0.5,sloped] {$\subseteq$};
\path (st) -- (at) node[pos=0.5,sloped] {$\supseteq$};
\path (ch) -- (st) node[pos=0.5,sloped] {$\subseteq$};
\end{tikzpicture}
\caption{GP\,2 label subtype hierarchy.}\label{fig:subtypes}
\end{figure}

The principal programming construct in GP\,2 are conditional graph transformation rules labelled with expressions. 
The formal definition of rules makes use of the following subset of list expressions, which specifies what expressions may occur in the left-hand graph of a rule.

\begin{defn}[Simple Expression]
\label{def:simple_expression}
 A list expression $e$ is \emph{simple}\/ if (1) $e$\/ contains no arithmetic, degree or length operators (with the possible exception of a unary minus preceding a sequence of digits), (2) $e$\/ contains at most one occurrence of a list variable, and (3) each occurrence of a string expression in $e$\/ contains at most one occurrence of a string variable.
\end{defn}

We call a label \emph{simple} if its list expression is simple. We can now define GP\,2's conditional graph transformation rules.

\begin{defn}[Rule]
\label{def:rule}
A \emph{rule} $\tuple{L \gets K \to R,\, c}$ consists of totally labelled rule graphs $L$\/ and $R$, a graph $K$\/ consisting of unlabelled nodes only, where $V_K \subseteq V_L$\/ and $V_K \subseteq V_R$, and a condition $c$ according to the grammar of Figure \ref{fig:condition-grammar}. Graph $K$\/ is called the \emph{interface}\/ of the rule. We require that all labels in $L$\/ are simple and that all variables occurring in $R$\/ also occur in $L$. In addition, all variables occurring in $c$\/ must occur in $L$.
\end{defn}

When a rule is represented graphically, the interface $K$\/ is given implicitly by the node identifiers that occur in $L$\/ and $R$\/ (which must coincide). We use small numbers attached to nodes as identifiers. For example, in the rule \texttt{link} of Figure \ref{fig:transitive-closure-program}, $K$\/ consists of three unlabelled nodes with identifiers \texttt{1}, \texttt{2} and \texttt{3}. In general, interface nodes are to be preserved, nodes in $L-K$\/ are to be deleted and nodes in $R-K$\/ are to be created. Node labels such as \texttt{x}, \texttt{y} and \texttt{z} are written inside nodes, whereas small integers below nodes are their identifiers. 
The constraint that all variables in $R$\/ must also occur in $L$\/ ensures that for a given match of $L$\/ in a host graph, applying the rule produces a unique graph (up to isomorphism). Similarly, the evaluation of the condition $c$ has a unique result if all its variables occur in $L$.

Next we introduce premorphisms which locate a rule's left-hand graph in a host graph.

\begin{defn}[Premorphism]
\label{def:premorphism}
 Given graphs $L$\/ and $G$, a \emph{premorphism} $g\colon L \to G$\/ consists of two functions $g_V\colon V_L \to V_G$ and $g_E\colon E_L \to E_G$\/ that preserve sources and targets, that is, $s_G \circ g_E = g_V \circ s_L$ and $t_G \circ g_E = g_V \circ t_L$. 
\end{defn}

We also need to introduce assignments, which instantiate variables in labels with host graph values of the same type.

\begin{defn}[Assignment]
\label{def:assignment}
An \emph{assignment}\/ is a family of mappings $$\alpha = (\alpha_X)_{X \in \{\mrm{L},\mrm{A},\mrm{I},\mrm{S},\mrm{C}\}}$$ where $\alpha_{\mrm{L}}\colon \mrm{LVar} \to \mrm{HostList}$, $\alpha_{\mrm{A}}\colon \mrm{AVar} \to \mrm{HostAtom}$, $\alpha_{\mrm{I}}\colon \mrm{IVar} \to \mrm{HostInteger}$, $\alpha_{\mrm{S}}\colon \mrm{SVar} \to \mrm{HostString}$,  and $\alpha_{\mrm{C}}\colon \mrm{CVar} \to \mrm{HostChar}$. We may omit the subscripts of these mappings as exactly one of them is applicable to a given variable.
\end{defn}

Given an expression $e \in \mrm{List}$, an assignment $\alpha$ and a premorphism $g\colon L \to G$\/ from a rule graph $L$\/ to a host graph $G$, the \emph{instance}\/ of $e$ with respect to $g$ and $\alpha$ is obtained by (1) replacing each variable $x$ with $\alpha(x)$, (2) replacing each node identifier $n$ with $g_V(n)$, and (3) evaluating the resulting expression according to the meaning of the operators in the grammar of Figure \ref{fig:rule-label-grammar}. The result is an expression in HostList which we denote by $e^{g,\alpha}$. 

For example, if $e = \mtt{x:1:x/2}$ and $\alpha(\mtt{x}) = 3$ then $e^{g,\alpha} = 3:1:1$ (note that '/' is integer division). If $e = \mtt{indeg(1)+outdeg(1)}$ then $e^{g,\alpha}$ is the sum of the indegree and outdegree of $g_V(\mtt{1})$ in $G$.

The evaluation of a condition $c$ with respect to $g$ and $\alpha$ proceeds analogously to the evaluation of list expressions but results in a value $c^{g,\alpha}$ which is either `true' or `false'. For example, if $c = \text{$\mtt{int(w)}$ \ttt{or} $\mtt{edge(1,2)}$}$ and $\alpha(\mtt{w}) = \mtt{1:2}$ then $c^{g,\alpha} = \mrm{true}$ if and only if there is an edge in $G$ from $g_V(\mtt{1})$ to $g_V(\mtt{2})$ (because $\mtt{int(w)}^{g,\alpha} = \mrm{false}$).

In the definitions below, given a node or edge $u$ in a graph $M$, we write $\mrm{mark}_M(u)$ for the mark of $u$ (if it exists). 

\begin{defn}[Compatibility]
\label{def:compatibility}
 A premorphism $g\colon L \to G$\/ and an assignment $\alpha$ are \emph{compatible}\/ if for all nodes and edges $u$ in $L$, (1) $l_L(u) = e \in \mrm{List}$ implies $l_G(g(u)) = e^{g,\alpha}$, (2) $l_L(u) = em$ with $e \in \text{List}$ and $m \in \mrm{HostMark}$ implies $l_G(g(u)) = e^{g,\alpha}m$, and (3) $l_L(u) = e\,\texttt{any}$ with $e \in \text{List}$ implies $l_G(g(u)) = e^{g,\alpha}m$ for some $m \in \mrm{HostMark}$.
 \end{defn}


\begin{defn}[Rule Application]
\label{def:rule_application}
A rule $r = \tuple{L \gets K \to R,\, c}$ is applied to a host graph $G$ as follows:
\begin{enumerate}
\item Find an injective premorphism $g\colon L \to G$\/ and an assignment $\alpha$ such that $g$ and $\alpha$ are compatible.
\item Check that $g$ satisfies the \emph{dangling condition}: no node in $g_V(V_L - V_K)$ is incident to an edge in $E_G - g_E(E_L)$.
\item Check that $c^{g,\alpha} = \mrm{true}$.
\item Construct a host graph $H$\/ from $G$\/ as follows: 
\begin{enumerate}
\item Remove all edges in $g_E(E_L)$ and all nodes in $g_V(V_L) - g_V(V_K)$, obtaining a subgraph $D$ of $G$.
\item $V_H = V_D + (V_R - V_K)$ and $E_H = E_D + E_R$. 
\item For each edge $e \in E_D$, $s_H(e) = s_D(e)$. For each edge $e \in E_R$, $s_H(e) = s_R(e)$ if $s_R(e) \in V_R-V_K$, otherwise $s_H(e) = g_V(s_R(e))$. Targets are defined analogously. 
\item For each edge $e \in E_D$, $m_H(e) = m_D(e)$. For each edge $e \in E_R$, $m_H(e) = \ul{\mrm{if}}\ m_R(e) \in \mrm{List}\ \ul{\mrm{then}}\ m_{R}(e)^{g,\alpha}\ \ul{\mrm{else}}\ \ul{\mrm{if}}\ m_R(e) = dm\ \text{with $m \in \mrm{Hostmark}$}\ \ul{\mrm{then}}\ d^{g,\alpha}m$.
\item For each node $v \in V_D - g_V(V_K)$, $l_H(v) = l_D(v)$.  For each node $v \in V_R - V_K$, $l_H(v) = \ul{\mrm{if}}\ l_R(v) \in \mrm{List}\ \ul{\mrm{then}}\ l_R(v)^{g,\alpha}\ \ul{\mrm{else}}\ \ul{\mrm{if}}\ l_R(v) = dm\ \text{with $m \in \mrm{Hostmark}$}\ \ul{\mrm{then}}\ d^{g,\alpha}m$. For each node $v \in g_V(V_K)$ with $g_V(\bar{v}) = v$, $l_H(v) = \ul{\mrm{if}}\ l_R(\bar{v}) \in \mrm{List}\ \ul{\mrm{then}}\ l_R(\bar{v})^{g,\alpha}\ \ul{\mrm{else}}\ \ul{\mrm{if}}\linebreak[3]  l_R(\bar{v}) = dm\ \text{with $m \in \mrm{Hostmark}$}\ \ul{\mrm{then}}\ d^{g,\alpha}m\ \ul{\mrm{else}}\ \ul{\mrm{if}}\ l_R(\bar{v}) = d\,\mtt{any}\ \ul{\mrm{then}}\ d^{g,\alpha}\mrm{mark}_G(v)$. 
\end{enumerate}
\end{enumerate}
We write $G \dder_{r,g} H$\/ to express that $H$\/ results from $G$ by applying $r$ with \emph{match}\/ $g$.
\end{defn}

Alternatively, rule application can be defined as a two-stage process in which first the rule $\tuple{L \gets K \to R,\, c}$ is instantiated to $\tuple{L^{g,\alpha} \gets K \to R^{g,\alpha},\, c^{g,\alpha}}$ which is a conditional rule in the double-pushout approach with relabelling \cite{Habel-Plump02c}. In the second stage, the instantiated rule is applied to the host graph by constructing two suitable pushouts in the category of partially labelled graphs. See \cite{Coutehoute23a} for details.

\subsection{Programs}
\label{subsec:programs}

Figure \ref{fig:command-syntax} defines the abstract syntax of GP\, 2 programs. We omit the syntax of rule declarations, rule identifiers and procedure identifiers, which can be found in \cite{Bak15a}. A program consists of declarations of conditional rules and procedures, and exactly one declaration of a main command sequence. The category \ttt{RuleId} refers to declarations of conditional rules in \ttt{RuleDecl}. Procedures must be non-recursive, they can be seen as macros with local declarations.

\begin{figure}[!h]
\centering
\begin{tabular}{lcl}
Prog & ::= & Decl \{Decl\} \\
Decl & ::= & RuleDecl $\mid$ ProcDecl $\mid$ MainDecl \\
ProcDecl & ::= & ProcId `=' [ `[' LocalDecl `]' ] ComSeq \\
LocalDecl & ::= & (RuleDecl $\mid$ ProcDecl) \{LocalDecl\} \\
MainDecl & ::= & \ttt{Main} `=' ComSeq \\
ComSeq & ::= & Com \{`;' Com\} \\
Com & ::= & RuleSetCall $\mid$ ProcCall \\
&& $\mid$ \ttt{if} ComSeq \ttt{then} ComSeq [\ttt{else} ComSeq] \\
&& $\mid$ \ttt{try} ComSeq [\ttt{then} ComSeq] [\ttt{else} ComSeq] \\
&& $\mid$ ComSeq `{!}' $\mid$ ComSeq \ttt{or} ComSeq \\
&& $\mid$ `(' ComSeq `)' $\mid$ \ttt{break} $\mid$ \ttt{skip} $\mid$ \ttt{fail} \\
RuleSetCall & ::= & RuleId $\mid$ `\{' [RuleId \{`,' RuleId\}] `\}' \\
ProcCall & ::= & ProcId 
\end{tabular}
\caption{Abstract syntax of GP\,2 programs. \label{fig:command-syntax}}
\end{figure}

\noindent 
Next, we describe the effect of GP\,2's main control constructs informally. 

The call of a rule set $\{r_1,\dots,r_n\}$ non-deterministically applies one of the rules whose left-hand graph matches a subgraph of the \emph{host graph} such that the dangling condition and the rule's application condition are satisfied. The call \emph{fails} if none of the rules is applicable to the host graph. 

The command \ttt{if} $C$ \ttt{then} $P$ \ttt{else} $Q$ is executed on a host graph $G$ by first executing $C$ on a copy of $G$. If this results in a graph, $P$ is executed on the original graph $G$; otherwise, if $C$ fails, $Q$ is executed on $G$. The \ttt{try} command has a similar effect, except that $P$ is executed on the result of $C$'s execution. 

The loop command $P!$ executes the body $P$ repeatedly until it fails. When this is the case, $P!$ terminates with the graph on which the body was entered for the last time. The \ttt{break} command inside a loop terminates that loop and transfers control to the command following the loop.

In general, the execution of a program on a host graph may result in different graphs, fail, or diverge. For a formal account of GP\,2's operational semantics, see \cite{Courtehoute-Plump21b}.

\subsection{Rooted Programs}
\label{subsec:rooted_programs}

The bottleneck for efficiently implementing a language based on graph transformation rules is the cost of graph matching. In general, to match the left-hand graph $L$ of a rule within a host graph $G$ requires time polynomial in the size of $L$ \cite{Bak-Plump12a,Bak-Plump16a}. As a consequence, linear-time graph algorithms in imperative languages may be slowed down to polynomial time when they are recast as rule-based programs.

To speed up matching, GP\,2 supports \emph{rooted} graph transformation where graphs in rules and host graphs are equipped with so-called root nodes. Roots in rules must match roots in the host graph so that matches are restricted to the neighbourhood of the host graph's roots. We draw root nodes using double circles.

A rule $\tuple{L \gets K \to R,\, c}$ is \emph{fast} if (1) each node in $L$ is undirectedly reachable from some root, (2) neither $L$ nor $R$ contain repeated occurrences of list, string or atom variables, and (3) the condition $c$ contains neither an $\mathtt{edge}$ predicate nor a test $e_1 \mtt{=} e_2$ or $e_1 \mtt{!\!=} e_2$ where both $e_1$ and $e_2$ contain a list, string or atom variable.


In \cite{Bak-Plump12a,Bak-Plump16a}, it is shown that rooted graph matching can be implemented to run in constant time for fast rules, provided there are upper bounds on the maximal node degree and the number of roots in host graphs. However, in Section \ref{s:enhancement} we show how to overcome the restriction that node degrees must be bounded.

\section{The GP\,2-to-C Compiler}
We now introduce the GP\,2-to-C compiler which translates GP\,2 programs into C code. A description of the compiler's implementation is beyond the scope of this paper; the interested reader may consult \cite{Bak15a,Bak-Plump16a,Campbell-Romo-Plump20d}. We also give a list of assumptions on the performance of some basic operations generated by the compiler. These assumptions are needed to analyse the time complexity of programs.

\subsection{Execution Pipeline}
\label{subsec:compiler}

Figure \ref{fig:gp2-to-c-compiler} shows that GP\,2 programs are first compiled to C programs using the GP\,2-to-C compiler and then transformed by GCC into platform-dependent binaries which are executed by the runtime system. The execution entails reading an input graph from the disc. 
 
\begin{figure}[!h]
\centering
\begin{tikzpicture}
\node at (1,1.7) {GP\,2 Program};
\draw [->] (1,1.5) -- (1,1);
\draw [fill=black!10] (0,0) rectangle node{GP\,2-to-C} (2,1);
\node at (2.75,0.75) {C};
\node at (2.75,0.25) {Program};
\draw [->] (2,0.5) -- (3.5,0.5);
\draw [fill=black!10] (3.5,0) rectangle node{GCC} (5.5,1);
\node at (6.25,0.75) {Binary};
\draw [->] (5.5,0.5) -- (7,0.5);
\node at (8.5,1.7) {Input Graph};
\draw [->] (8.5,1.5) -- (8.5,1);
\draw [fill=black!10] (7,0) rectangle node{Runtime System} (10,1);
\draw [->] (10,0.5) -- (11.5,0.5);
\node at (10.75,0.75) {Output};
\node at (10.75,0.25) {Graph};
\end{tikzpicture}
\caption{GP\,2 execution pipeline.}
\label{fig:gp2-to-c-compiler}
\end{figure}

Executing a program results in one of three possible outcomes: the program (1) produces an output graph, (2) evaluates to \ttt{fail}, or (3) diverges (does not terminate).


For  example, consider the GP\,2 program in Figure \ref{fig:transitive-closure-program} which computes the transitive closure of a graph.\footnote{The concrete syntax of an equivalent program with unmarked nodes can be found at \url{https://github.com/UoYCS-plasma/GP2/blob/master/programs/transitive-closure.gp2}.} The C code generated for this program (1) reads and parses the input host graph, (2) applies the \ttt{link} rule as long as possible, and (3) returns the resulting host graph as output.

\begin{figure}[!h]
\centering
\fbox{\begin{minipage}{11.78cm}
\hspace{1em}\texttt{Main = link!}

\setlength{\tabcolsep}{0.4cm}

\medskip
\smallskip

\begin{tabular}{ p{10.5cm} }
	
	\ttt{link(a,b,x,y,z:list)} \\
	
	\begin{tikzpicture}
        \tikzstyle{every node}=[font=\ttfamily]
        
		\node (a) at (0,0)       [circle, thick, rounded corners=7, draw=black, minimum size=5mm, fill=gray!50, label=below:\tiny\tiny1] {x};
		\node (b) at (1.25,0)    [circle, thick, rounded corners=7, draw=black, minimum size=5mm, fill=gray!50, label=below:\tiny\tiny2] {y};
		\node (c) at (2.5,0)     [circle, thick, rounded corners=7, draw=black, minimum size=5mm, fill=gray!50, label=below:\tiny\tiny3] {z};
		
		\node (d) at (3.375,0)   {$\Rightarrow$};
		
		\node (e) at (4.25,0)    [circle, thick, rounded corners=7, draw=black, minimum size=5mm, fill=gray!50, label=below:\tiny\tiny1] {x};
		\node (f) at (5.5,0)     [circle, thick, rounded corners=7, draw=black, minimum size=5mm, fill=gray!50, label=below:\tiny\tiny2] {y};
		\node (g) at (6.75,0)    [circle, thick, rounded corners=7, draw=black, minimum size=5mm, fill=gray!50, label=below:\tiny\tiny3] {z};
		
		\draw (a) edge[->,thick] node[above] {a} (b)
		      (b) edge[->,thick] node[above] {b} (c)
		      (e) edge[->,thick] node[above] {a} (f)
		      (f) edge[->,thick] node[above] {b} (g)
		      (e) edge[->,thick,bend left=40] (g);
	\end{tikzpicture}
	\\
	
	\vspace{-1em}\ttt{where not edge(1,3)} \\
	
\end{tabular}
\end{minipage}
}
\caption{GP\,2 program \ttt{transitive-closure}.}
\label{fig:transitive-closure-program}
\end{figure}

Applying the \ttt{link} rule as long as possible requires to find a match satisfying the application condition, and if one is found, adding a fresh edge according to the rule and the match. This is repeated zero or more times until a match cannot be found anymore. Finding a match involves finding a grey node corresponding to node \texttt{1}, then finding an outgoing edge to another grey node, corresponding to node \texttt{2}, then finding an outgoing edge to yet another grey node, corresponding to node \texttt{3}, and then checking the application condition. This search is controlled by the matching algorithm and is described in general in Bak's thesis \cite{Bak15a}. Updates to the algorithm are described in \cite{Campbell-Courtehoute-Plump19b}).

We conclude this subsection by stressing that the GP\,2-to-C compiler does not attempt to avoid program failure or to generate more than one output graph. This would require backtracking, which is only used in the rule matching algorithm and in the implementation of rule sets (a rule set only fails if none of its rules can be matched). The compiler has been deliberately designed not to backtrack in program execution, which is necessary to achieve the fast run times reported in this paper.

\subsection{Reasoning about Time Complexity}
\label{subsec:complexity}

When analysing the time complexity of programs, we assume that these are fixed. This is customary in algorithm analysis where programs are fixed and runtime is measured in terms of input size \cite{Aho-Hopcroft-Ullman74a,Skiena20a}. In our setting, the input size is the \emph{size} of a host graph, which we define to be the total number of nodes and edges. 

When we discuss the time complexity of rule applications, it suffices to reason about the complexity of finding a match because all the programs in this paper satisfy the assumption of the following lemma.

\begin{lemC}[Constant Time Graph Construction\cite{campbell2022fast}] \label{lem:const-time}
    If a match has been found for a rule that does not modify labels, other than possibly changing marks or changing non-host-graph labels to host-graph labels, the resulting graph can be constructed in constant time.
\end{lemC}
\noindent 
To reason about a program's time complexity, we need to make assumptions on the complexity of certain elementary operations generated by the compiler. Figure \ref{fig:complexity-assumptions} shows the complexity of the basic procedures of the rule matching algorithm, adapted from \cite{campbell2022fast}. The grey rows indicate procedures updated by the changes introduced in Section~\ref{s:enhancement}. The runtimes reported in this paper's case studies and in \cite{ismaili2024linear,Ismaili-Plump25a} are consistent with the complexity assumptions of Figure \ref{fig:complexity-assumptions}.

\definecolor{modified}{gray}{0.9}
\definecolor{added}{rgb}{0.88,1,1}

\begin{figure}[!h]
\begin{center}
\scalebox{.7}{
\begin{tabular}{llr} 
Operation               & Description                                                    & Complexity  \\ \hline
\texttt{alreadyMatched}    & Test if the given item has been matched in the host graph.     & \(\mathrm{O}(1)\)    \\
\texttt{clearMatched}      & Clear the \texttt{is matched} flag for a given item.              & \(\mathrm{O}(1)\)    \\
\texttt{setMatched}        & Set the \texttt{is matched} flag for a given item.                & \(\mathrm{O}(1)\)    \\
\rowcolor{modified}\texttt{firstHostNode(m)}     & \textbf{Fetch the first node of mark \texttt{m} in the host graph.}                        & \(\mathrm{O}(1)\)    \\
\rowcolor{modified}\texttt{nextHostNode(m)}      & \textbf{Given a node of mark \texttt{m}, fetch the next node of mark \texttt{m} in the host graph.}           & \(\mathrm{O}(1)\)    \\
\texttt{firstHostRootNode\,\,\,} & Fetch the first root node in the host graph.                   & \(\mathrm{O}(1)\)    \\
\texttt{nextHostRootNode\,\,\,}  & Given a root node, fetch the next root node in the host graph. & \(\mathrm{O}(1)\)    \\
\rowcolor{modified}\texttt{firstInEdge(m)}       & \textbf{Given a node, fetch the first inedge of mark \texttt{m}.}                 & \(\mathrm{O}(1)\)    \\
\rowcolor{modified}\texttt{nextInEdge(m)}        & \textbf{Given a node and an edge of mark \texttt{m}, fetch the next inedge of mark \texttt{m}.}\,\,\,         & \(\mathrm{O}(1)\)    \\
\rowcolor{modified}\texttt{firstOutEdge(m)}      & \textbf{Given a node, fetch the first outedge of mark \texttt{m}.}                   & \(\mathrm{O}(1)\)    \\
\rowcolor{modified}\texttt{nextOutEdge(m)}       & \textbf{Given a node and an edge of mark \texttt{m}, fetch the next outedge of mark \texttt{m}.}        & \(\mathrm{O}(1)\)    \\
\rowcolor{modified}\texttt{firstLoop(m)}       & \textbf{Given a node, fetch the first loop edge of mark \texttt{m}.}        & \(\mathrm{O}(1)\)    \\
\rowcolor{modified}\texttt{nextLoop(m)}       & \textbf{Given a node and an edge of mark \texttt{m}, fetch the next loop edge of mark \texttt{m}.}        & \(\mathrm{O}(1)\)    \\
\texttt{getInDegree}       & Given a node, fetch its indegree.                       & \(\mathrm{O}(1)\)    \\
\texttt{getOutDegree}      & Given a node, fetch its outdegree.                       & \(\mathrm{O}(1)\)    \\
\texttt{getMark}           & Given a node or edge, fetch its mark.                          & \(\mathrm{O}(1)\)    \\
\texttt{isRooted}          & Given a node, determine if it is rooted.                       & \(\mathrm{O}(1)\)    \\
\texttt{getSource}         & Given an edge, fetch the source node.                          & \(\mathrm{O}(1)\)    \\
\texttt{getTarget}         & Given an edge, fetch the target node.                          & \(\mathrm{O}(1)\)    \\
\texttt{parseInputGraph}   & Parse and load the input graph into memory.    & \(\mathrm{O}(n)\)    \\
\texttt{printHostGraph}    & Write the current host graph as output.                  & \(\mathrm{O}(n)\)    \\
\end{tabular}
}
\end{center}
\caption{Updated complexity assumptions on elementary operations generated by the compiler, where $n$ is the size of the input. Modified operations are highlighted in grey.}
\label{fig:complexity-assumptions}
\end{figure}

\section{The Problem with Unbounded-Degree Graphs}
\label{s:problem1}
This section explains the problem caused by classes of input graphs with an unbounded node degree, while the next section introduces the compiler improvements.

Previously, non-destructive GP\,2 programs based on depth-first search ran in linear time on graph classes of bounded node degree but in non-linear time on graph classes of unbounded degree \cite{campbell2022fast}. For example, consider the program \texttt{is-connected} in Figure \ref{fig:is-con-fig-old} which checks whether a graph is connected.\footnote{A (possibly dashed) line without arrowhead between two nodes is a \emph{bidirectional}\/ edge which matches host graph edges of either direction.} Input graphs are arbitrary \gp{} host graphs with grey, non-rooted nodes and unmarked edges. If the graph is connected, the program returns a graph that is isomorphic up to marks. On disconnected graphs, the program fails.
\begin{figure}[!ht]
\begin{mdframed}[linewidth=0.8pt]
\begin{verbatim}
Main  = try init then (DFS!; Check)
DFS   = forward!; try back else break
Check = if match then fail

\end{verbatim}

\begin{tikzpicture}
\tikzstyle{every node}=[font=\ttfamily]

\draw (0.9,0.75) node[align=left] {init(x:list)};
\node[rectangle, thick, thick, rounded corners=7, draw=black, minimum size=5mm, fill=gray!50, label=below:\tiny\tiny 1](a2) at (0,0){x};
\draw (1,0) node[] {$\Rightarrow$};
\node[rectangle, thick, rounded corners=7, draw=black, minimum size=5mm,  double, double distance=1pt, fill=cyan, label=below:\tiny\tiny1](a2) at (2,0){x};
\draw[->, thick];
\draw (4.8,1) -- (4.8,-0.5);
\end{tikzpicture}
\begin{tikzpicture}
\tikzstyle{every node}=[font=\ttfamily]
\hspace{1.55em}
\draw (1,0.75) node[align=left] {match(x:list)};
\node[rectangle, thick, rounded corners=7, draw=black, minimum size=5mm, fill=gray!50, label=below:\tiny\tiny1](a2) at (0,0){x};
\draw (1,0) node[] {$\Rightarrow$};
\node[rectangle, thick, rounded corners=7, draw=black, minimum size=5mm, fill=gray!50, label=below:\tiny\tiny 1](a2) at (2,0){x};
\draw[->, thick];

\end{tikzpicture}

\begin{tikzpicture}
\tikzstyle{every node}=[font=\ttfamily]
\draw (1.6,0.75) node[align=left] {forward(x,y,z:list)};
\node[rectangle, thick, rounded corners=7, draw=black, minimum size=5mm,  double, double distance=1pt, fill=cyan, label=below:\tiny\tiny1](a1) at (0,0){x};
\node[rectangle, thick, rounded corners=7, draw=black, minimum size=5mm, fill=gray!50, label=below:\tiny\tiny2](a2) at (1,0){y};
\draw (2,0) node[] {$\Rightarrow$};
\node[rectangle, thick, rounded corners=7, draw=black, minimum size=5mm, fill=cyan, label=below:\tiny\tiny1](a3) at (3,0){x};
\node[rectangle, thick, rounded corners=7, draw=black, minimum size=5mm, double, double distance=1pt, fill=cyan, label=below:\tiny\tiny2](a4) at (4,0){y};
\draw[-, line width=1.2pt] (a1) edge[] node[above, color = black]{z} (a2) (a3) edge[dashed] node[above, color = black]{z} (a4);
\draw (4.8,1) -- (4.8,-0.5);
\end{tikzpicture}
\hspace{1.5em}
\begin{tikzpicture}
\tikzstyle{every node}=[font=\ttfamily]
\draw (1.3,0.75) node[align=left] {back(x,y,z:list)};
\node[rectangle, thick, rounded corners=7, draw=black, minimum size=5mm, fill=cyan, label=below:\tiny\tiny1](a1) at (0,0){x};
\node[rectangle, thick, rounded corners=7, draw=black, minimum size=5mm, double, double distance=1pt, fill=cyan, label=below:\tiny\tiny2](a2) at (1,0){y};
\draw (2,0) node[] {$\Rightarrow$};
\node[rectangle, thick, rounded corners=7, draw=black, minimum size=5mm, double, double distance=1pt, fill=cyan, label=below:\tiny\tiny1](a3) at (3,0){x};
\node[rectangle, thick, rounded corners=7, draw=black, minimum size=5mm, label=below:\tiny\tiny2](a4) at (4,0){y};
\draw[-, line width=1.2pt] (a1) edge[dashed,shorten <=1.5] node[above, color = black]{z} (a2) (a3) edge[] node[above, color = black]{z} (a4);
\end{tikzpicture}

\end{mdframed}
\caption{The old program \texttt{is-connected}.}
\label{fig:is-con-fig-old}
\end{figure}

The rule \texttt{init} picks an arbitrary grey node as a root (if the input graph is non-empty) and then the loop \texttt{DFS!} performs a depth-first search of the connected component of the node chosen by \texttt{init}. Rule \texttt{forward} marks each newly visited node blue, and \texttt{back} unmarks it once it has been processed. The procedure \texttt{DFS} ends when \texttt{back} fails to match, indicating that the search is complete. Rule \texttt{match} checks whether a grey-marked node still exists in the graph following the execution of \texttt{DFS!}. This is the case if and only if the input graph contains more than one connected component. In this situation the program invokes the command \texttt{fail}, otherwise it terminates by returning the graph resulting from the depth-first search.

It can be shown that the program \texttt{is-connected} runs in linear time on classes of graphs with bounded node degree \cite{campbell2022fast}. However, as the following example shows, the program may require non-linear time on unbounded-degree graph classes. Figure \ref{fig:unbounded} shows an execution of \texttt{is-connected} on a star graph with $8$ edges. Below each graph, we record the range of attempts that the matching algorithm may make. For instance, in the second graph of the top row, either a match is found immediately among the edges that connect the central node with the grey nodes, or the dashed edge is unsuccessfully tried first. 
\begin{figure}[!ht]
    \centering
        \input{Figures/star-unbounded}
    \caption{Matching attempts with the \texttt{forward} rule. \texttt{fd} and \texttt{bk} denote \texttt{forward} and \texttt{back}, respectively.}
    \label{fig:unbounded}
\end{figure}
In order to find a match for the rule \texttt{forward}, the matching algorithm considers, in the worst case, every edge incident with the root. When the node central to the graph is rooted and the rule \texttt{forward} is called, the matching algorithm may first attempt a match with the dashed back edge and all edges incident with an unmarked node. Therefore, the maximum number of matching attempts for \texttt{forward} grows as the root moves back to the central node. As can be seen from this example, the worst-case complexity of matching \texttt{forward} throughout the program's execution is $2|E|+\sum_{i=1}^{|E|}i=\mathrm{O}(|E|^2)$ where $E$ is the set of edges.


\section{Compiler Enhancements}
\label{s:enhancement}
\newcommand{\red}[1]{\textcolor{red}{#1}}
\newcommand{\blue}[1]{\textcolor{blue}{#1}}
\newcommand{\green}[1]{\textcolor{green}{#1}}
\newcommand{\brown}[1]{\textcolor{brown}{#1}}
\newcommand{\black}[1]{\textcolor{black}{#1}}

To overcome the limitation that depth-first search requires bounded-degree graph classes to run in linear time, we changed the GP\,2 compiler described in \cite{Campbell-Romo-Plump20d}. We refer to that compiler as the \emph{2020 compiler} and call the version introduced in this paper the \emph{new compiler}\footnote{Available at: \url{https://github.com/UoYCS-plasma/GP2}.}.

\subsection{First Enhancement}
\label{ss:first-enhancement}

The 2020 compiler stored the host graph's structure as one linked list containing every node in the graph, with each node storing two linked lists of edges: one for incoming edges and one for outgoing edges. Searching through these lists to find an edge with a particular mark could involve looking up a large number of edges with incompatible marks, slowing down the search. In particular, the number of unmatchable edges could grow during program execution if the marks changed.


When trying to match an edge of a given rule, host graph edges with incorrect orientation or mark can be ignored. Thus, by organising the edges incident to a node into linked lists based on their mark and orientation, the matching algorithm can select edges of correct mark and orientation. More precisely, the new compiler generates for each node a two-dimensional array such that each entry is a linked list of all incident edges with a particular combination of mark and orientation. We also consider loops to be a distinct type of orientation. The array, therefore, consists of five rows (unmarked, dashed, red, green, blue) and three columns (incoming, outgoing, loop), totalling $15$ cells. See Figure \ref{fig:enter-label} for an illustration. We mention that the program in Figure~\ref{fig:is-con-fig-old} still does not run in linear time for input graph classes of unbounded degree; see Section~\ref{s:connectedness} for an adjusted program which exploits the aforementioned changes and runs in linear time.

Note that, given a rule with a node $v$ and an incident edge with a particular mark and orientation, the matching algorithm can access in constant time the corresponding list in $v$'s array. The first element of that list is a host graph edge with matching mark and orientation; in case the list is empty, the algorithm has determined that no matching edge exists in the host graph. 

\begin{figure}[htb] 
    \begin{center}
    \begin{tabular}{l|c|c|c}
                & in   & out    & loop \\
  \hline
  unmarked      & \dots & \dots & \dots \\
  \hline
  dashed        & \dots & \dots & \dots \\
  \hline
  \red{red}     & \dots & \dots & \dots \\
  \hline
  \green{green} & \dots & \dots & \dots \\
  \hline
  \blue{blue}   & \dots & \dots & \dots \\

\end{tabular} 
\end{center}
    \caption{Two-dimensional array of linked lists of edges.}
    \label{fig:enter-label}
\end{figure}

\subsection{Finding Nodes in Constant Time}
\label{s:problem2}
The previous subsection implies that the code generated by the new compiler can match edges labelled with variables of type \ttt{list} in constant time (or detect that there are no matching edges in the host graph). The corresponding problem of finding nodes labelled with \ttt{list} variables in constant time could not be solved with the 2020 compiler either. This is explained below. 

Consider the program \texttt{is-discrete} in Figure \ref{fig:is-discrete} which checks whether a graph is discrete, that is, contains no edges. Given an unmarked input graph, the program fails if and only if the input graph contains any edges. The program repeatedly picks some unmarked node, turns it into a red root node, and checks whether the node is isolated by applying the rule \ttt{isolated}. The left-hand node of the rule is not an interface node and hence, by the dangling condition, the rule is applicable if and only if the root node is isolated. If \texttt{isolated} is applied, the root node is unrooted and the loop continues. If \texttt{isolated} is not applicable, a non-isolated node has been found and the loop breaks. In this case rule \texttt{root} finds that a root exists and the \ttt{fail} command is invoked.

\definecolor{gp2pink}{RGB}{255, 153, 238}

\begin{figure}[!ht]
    \begin{mdframed}[linewidth=1pt]
    \begin{verbatim}
Main = (mark; try isolated else break)!; if root then fail
    \end{verbatim}
    
    \begin{tikzpicture}
    \tikzstyle{every node}=[font=\ttfamily]
    
    \draw (0.9,0.75) node[align=left] {mark(x:list)};
    \node[rectangle, thick, thick, rounded corners=7, draw=black, minimum size=5mm, label=below:\tiny\tiny 1](a2) at (0,0){x};
    \draw (1,0) node[] {$\Rightarrow$};
    \node[rectangle, thick, rounded corners=7, draw=black, minimum size=5mm,  double, double distance=1pt, fill=red!60, label=below:\tiny\tiny1](a2) at (2,0){x};
    \draw[->, thick];
    \draw (3.2,1) -- (3.2,-0.5);
    \end{tikzpicture}
    \hspace{1em}
    \begin{tikzpicture}
    \tikzstyle{every node}=[font=\ttfamily]
    \draw (1.3,0.75) node[align=left] {isolated(x:list)};
    \node[rectangle, thick, rounded corners=7, draw=black, minimum size=5mm,  double, double distance=1pt, fill=red!60, label=below:\tiny\tiny](a2) at (0,0){x};
    \draw (1,0) node[] {$\Rightarrow$};
    \node[rectangle, thick, rounded corners=7, draw=black, minimum size=5mm, fill=red!60, label=below:\tiny\tiny](a2) at (2,0){x};
    \draw (3.4,1) -- (3.4,-0.5);
    \draw (0,-0.62) -- (0,-0.62);
    \draw[->, thick];
    \end{tikzpicture}
    \hspace{1em}
    \begin{tikzpicture}
    \tikzstyle{every node}=[font=\ttfamily]
     \draw (0.9,0.75) node[align=left] {root(x:list)};
    \node[rectangle, thick, thick, rounded corners=7, draw=black, minimum size=5mm, fill=red!60,  double, double distance=1pt,  label=below:\tiny\tiny 1](a2) at (0,0){x};
    \draw (1,0) node[] {$\Rightarrow$};
    \node[rectangle, thick, rounded corners=7, draw=black, minimum size=5mm, double, double distance=1pt, fill=red!60, label=below:\tiny\tiny1](a2) at (2,0){x};
    \draw[->, thick];
    \end{tikzpicture}
    \end{mdframed}
    \caption{The program \texttt{is-discrete}.}
    \label{fig:is-discrete}
    \end{figure}

The code generated by the 2020 compiler matched the rule \texttt{mark} with a complexity of $\mathrm{O}(n)$, where $n$ is the number of nodes in the host graph. This is because finding an unmarked node required iterating through a unique linked list holding all nodes in the host graph. As a result, the overall time complexity of \texttt{is-discrete} was $\mathrm{O}(n^2)$, as empirically evidenced by the timing in Figure \ref{fig:bench-is-dis}.

\begin{figure}[!ht]
    \centering
    \begin{tikzpicture}
    \begin{axis}[
      xlabel=number of nodes,
      ylabel=runtime (ms), ylabel style={above=0.2mm},
      width=9.2cm,height=7.2cm,
      legend style={at={(1.525,0.5)}},
      ymajorgrids=true,
      grid style=dashed]
      \addplot[color=plot8, mark=square*] table [y=time, x=n]{Figures/Benchmarks/is-dis.dat};
      \addlegendentry{2020 compiler}
      \addplot[color=plot7, mark=square*] table [y=time, x=n]{Figures/Benchmarks/is-dis-new.dat};
      \addlegendentry{New compiler}
    \end{axis}  
    \end{tikzpicture}
    \caption{Measured performance of the program \texttt{is-discrete} on discrete graphs with the 2020 compiler and the new compiler.}
    \label{fig:bench-is-dis}
    \end{figure}

\subsection{Second Enhancement}
\label{ss:second-enhancement}
As explained above, the 2020 compiler generated code storing all host graph nodes in a single linked list. Hence the matching algorithm had to carry out  unnecessary lookups to find a match for a node in the host graph when there existed nodes with marks that did not match. 
The solution in the new compiler is to store nodes in five separate linked lists: one list for each node mark and one list for unmarked nodes. The linked lists are stored in an array, providing constant-time access to the first element of each list. See Figure \ref{fig:node-lists} for an illustration. This structure allows to match the rule \texttt{mark} in constant time as the matching algorithm immediately inspects the linked list holding all unmarked nodes in the host graph. Since the left-hand node of \texttt{mark} is labelled with a variable of type \ttt{list}, the first node found in the linked list by the matching algorithm matches successfully. As a result of the improved data structure, the time complexity of the program \texttt{is-discrete} under the new compiler is $\mathrm{O}(n)$. Figure \ref{fig:bench-is-dis} highlights the difference in runtimes of the program \texttt{is-discrete} run under both compilers.

\begin{figure}
\begin{center}
\begin{tabular}{l|c} 
              
  unmarked                  & \dots \\
  \hline
  \textcolor{gp2grey}{grey} & \dots \\
  \hline
  \red{red}                 & \dots \\
  \hline
  \green{green}             & \dots \\
  \hline
  \blue{blue}               & \dots 
  
\end{tabular}  
    \caption{Array of linked lists of nodes.}
    \label{fig:node-lists}
\end{center}
\end{figure}

\section{Case Study: Recognising Connected Graphs}
\label{s:connectedness}
Checking for the connectedness of a graph is a fundamental application of depth-first search (DFS). A \gp{} program solving this problem is given in \cite{Bak15a, campbell2022fast}, but to run in linear time it requires input graphs of bounded node degree.

\subsection{Program}
\label{ss:problem-con}

The program \texttt{is-connected}\footnote{The concrete syntax of the program available at: \url{https://gist.github.com/ismaili-ziad/7b419a12a8d156abf5ad8ac8c352713a}.} in Figure \ref{fig:is-con-fig} recognises connected graphs, in that each pair of nodes are connected by a sequence of edges of arbitrary orientation, with respect to the following specification.

\begin{description}
\item[\textbf{Input}] An arbitrary \gp{} host graph such that
\begin{enumerate}
        \item each node is non-rooted and marked grey, and
        \item each edge is unmarked.
    \end{enumerate}
\item[\textbf{Output}] If the input graph is connected, a host graph that is isomorphic to the input graph up to marks. Otherwise \emph{failure}.
\end{description}
\noindent 
This can be achieved by conducting a DFS from an arbitrarily chosen node while marking newly visited nodes. Since the DFS cannot propagate beyond the connected component it started in, the presence of an unmarked node following the DFS traversal indicates that the host graph is not connected. This program improves implementations of the same problem in \cite{Bak15a, campbell2022fast}, which is discussed in Section~\ref{s:problem1}, by achieving linear-time complexity on arbitrary input graphs (including star graphs). 

We stress that the old \texttt{is-connected} program (Figure~\ref{fig:is-con-fig-old}) still does not run in linear time for input graph classes of unbounded degree. This means that the new compiler does not necessarily improve the time complexity of existing programs. Instead, we need a new programming technique to exploit the improved graph data structure. To explain this technique, consider the program in Figure~\ref{fig:is-con-fig}. The rule \texttt{next\_edge} selects in constant time an adjacent node that shares an unmarked edge with the root. Subsequently, the ruleset \texttt{\{move, ignore\}} acts like a case statement, distinguishing the cases where the adjacent node is blue or grey. We use a similar technique for the case studies of Section~\ref{s:acyclicity} and Section~\ref{s:bf}.

\begin{center}
    \begin{figure}[!ht]
\begin{mdframed}[linewidth=0.8pt]
\begin{verbatim}
Main    = try init then (DFS!; Check)
DFS     = FORWARD!; try back else break
FORWARD = next_edge; {move, ignore}
Check   = if match then fail

\end{verbatim}

\begin{tikzpicture}
\tikzstyle{every node}=[font=\ttfamily]

\draw (0.9,0.75) node[align=left] {init(x:list)};
\node[rectangle, thick, thick, rounded corners=7, draw=black, minimum size=5mm, fill=gray!50, label=below:\tiny\tiny 1](a2) at (0,0){x};
\draw (1,0) node[] {$\Rightarrow$};
\node[rectangle, thick, rounded corners=7, draw=black, minimum size=5mm,  double, double distance=1pt, fill=cyan, label=below:\tiny\tiny1](a2) at (2,0){x};
\draw[->, thick];
\draw (4.8,1) -- (4.8,-0.5);
\end{tikzpicture}
\begin{tikzpicture}
\tikzstyle{every node}=[font=\ttfamily]
\hspace{1.55em}
\draw (1,0.75) node[align=left] {match(x:list)};
\node[rectangle, thick, rounded corners=7, draw=black, minimum size=5mm, fill=gray!50, label=below:\tiny\tiny1](a2) at (0,0){x};
\draw (1,0) node[] {$\Rightarrow$};
\node[rectangle, thick, rounded corners=7, draw=black, minimum size=5mm, fill=gray!50, label=below:\tiny\tiny 1](a2) at (2,0){x};
\draw[->, thick];

\end{tikzpicture}

\begin{tikzpicture}
\tikzstyle{every node}=[font=\ttfamily]
\draw (1.8,0.75) node[align=left] {next\_edge(x,y,z:list)};
\node[rectangle, thick, rounded corners=7, draw=black, minimum size=5mm,  double, double distance=1pt, fill=cyan, label=below:\tiny\tiny1](a1) at (0,0){x};
\node[rectangle, thick, rounded corners=7, draw=black, minimum size=5mm, fill=magenta!60, label=below:\tiny\tiny2](a2) at (1,0){y};
\draw (2,0) node[] {$\Rightarrow$};
\node[rectangle, thick, rounded corners=7, draw=black, double, double distance=1pt, minimum size=5mm, fill=cyan, label=below:\tiny\tiny1](a3) at (3,0){x};
\node[rectangle, thick, rounded corners=7, draw=black, minimum size=5mm, fill=magenta!60, label=below:\tiny\tiny2](a4) at (4,0){y};
\draw[-, line width=1.2pt] (a1) edge[] node[above, color = black]{z} (a2) (a3) edge[red] node[above, color = black]{z} (a4);
\draw (4.8,1) -- (4.8,-0.5);
\end{tikzpicture}
\hspace{1.35em}
\begin{tikzpicture}
\tikzstyle{every node}=[font=\ttfamily]
\draw (1.5,0.75) node[align=left] {ignore(x,y,z:list)};
\node[rectangle, thick, rounded corners=7, draw=black, minimum size=5mm,  double, double distance=1pt, fill=cyan, label=below:\tiny\tiny1](a1) at (0,0){x};
\node[rectangle, thick, rounded corners=7, draw=black, minimum size=5mm, fill=cyan, label=below:\tiny\tiny2](a2) at (1,0){y};
\draw (2,0) node[] {$\Rightarrow$};
\node[rectangle, thick, rounded corners=7, draw=black, double, double distance=1pt, minimum size=5mm, fill=cyan, label=below:\tiny\tiny1](a3) at (3,0){x};
\node[rectangle, thick, rounded corners=7, draw=black, minimum size=5mm, fill=cyan, label=below:\tiny\tiny2](a4) at (4,0){y};
\draw[-, line width=1.2pt] (a1) edge[red] node[above, color = black]{z} (a2) (a3) edge[cyan] node[above, color = black]{z} (a4);
\end{tikzpicture}

\begin{tikzpicture}
\tikzstyle{every node}=[font=\ttfamily]
\draw (1.35,0.75) node[align=left] {move(x,y,z:list)};
\node[rectangle, thick, rounded corners=7, draw=black, minimum size=5mm,  double, double distance=1pt, fill=cyan, label=below:\tiny\tiny1](a1) at (0,0){x};
\node[rectangle, thick, rounded corners=7, draw=black, minimum size=5mm, fill=gray!50, label=below:\tiny\tiny2](a2) at (1,0){y};
\draw (2,0) node[] {$\Rightarrow$};
\node[rectangle, thick, rounded corners=7, draw=black, minimum size=5mm, fill=cyan, label=below:\tiny\tiny1](a3) at (3,0){x};
\node[rectangle, thick, rounded corners=7, draw=black, minimum size=5mm, double, double distance=1pt, fill=cyan, label=below:\tiny\tiny2](a4) at (4,0){y};
\draw[-, line width=1.2pt] (a1) edge[red] node[above, color = black]{z} (a2) (a3) edge[dashed] node[above, color = black]{z} (a4);
\draw (4.8,1) -- (4.8,-0.5);
\end{tikzpicture}
\hspace{1.35em}
\begin{tikzpicture}
\tikzstyle{every node}=[font=\ttfamily]
\draw (1.3,0.75) node[align=left] {back(x,y,z:list)};
\node[rectangle, thick, rounded corners=7, draw=black, minimum size=5mm, fill=cyan, label=below:\tiny\tiny1](a1) at (0,0){x};
\node[rectangle, thick, rounded corners=7, draw=black, minimum size=5mm, double, double distance=1pt, fill=cyan, label=below:\tiny\tiny2](a2) at (1,0){y};
\draw (2,0) node[] {$\Rightarrow$};
\node[rectangle, thick, rounded corners=7, draw=black, minimum size=5mm, double, double distance=1pt, fill=cyan, label=below:\tiny\tiny1](a3) at (3,0){x};
\node[rectangle, thick, rounded corners=7, draw=black, minimum size=5mm, fill=cyan, label=below:\tiny\tiny2](a4) at (4,0){y};
\draw[-, line width=1.2pt] (a1) edge[dashed,shorten <=1.5] node[above, color = black]{z} (a2) (a3) edge[cyan] node[above, color = black]{z} (a4);
\end{tikzpicture}

\end{mdframed}
\caption{The (new) program \texttt{is-connected}.}
\label{fig:is-con-fig}
\end{figure}
\end{center}

\subsection{Proof of Correctness}

Let us examine the correctness of \texttt{is-connected}.
\noindent 
 We write $|X|$ for the cardinality of a finite set $X$. We use the notation $A \Rightarrow_{r} B$ to indicate that $B$ results from applying $r$ to $A$. 


First of all, we establish what we mean by an input graph in the context of the \texttt{is-connected} program.

\begin{defn}[Connectedness]
    A graph $G$ is \emph{connected} if and only if, for any pair of distinct vertices $u$ and $v$ in $G$, there exists a path of edges (regardless of direction) from $u$ to $v$.
\end{defn}

From now on, we call graphs satisfying the assumptions of the specification given at the beginning of the subsection as \emph{input graphs}.

The choice of having input nodes grey is motivated by the fact that nodes with the wildcard colour magenta (corresponding to \texttt{any}) cannot match unmarked nodes. We overcome this problem by requiring that all nodes in input graphs are marked grey.

\begin{prop}
\label{prop:all-same-comp-con}
    Throughout the execution of \texttt{is-connected} on an input graph all non-grey nodes in the host graph share a connected component.
\end{prop}
\begin{proof}
    All nodes in the input graph are initially grey. Upon inspection of the rules, \texttt{init} and \texttt{move} are the only rules altering the mark of a grey node, where \texttt{init} is only called once at the start of the program. Let us show inductively that the invariant holds.

    If \texttt{init} is applied at the start of the program, the host graph will contain exactly one non-grey node, which does not violate the invariant. If \texttt{init} is not applied, \texttt{move} is never called, and the invariant remains true.

    Now, assume that the invariant is true at some stage of the execution. If \texttt{move} is applied, the non-grey node $v$ that turns blue has to be adjacent to a non-grey node. Thus, if all non-grey nodes shared a connected component prior to the application of \texttt{move}, $v$ has to be part of that component as well. Thus, the invariant remains satisfied.
\end{proof}
\begin{lem}[Termination of 
\texttt{is-connected}]
\label{lem:term-con}
    The program \texttt{is-connected} terminates on any input graph.
\end{lem}
\begin{proof}
    The program \texttt{is-connected} contains two looping bodies: \texttt{DFS!} and \texttt{FORWARD!}. We first show termination of \texttt{FORWARD!}. Consider the measure $\#_1X$ consisting of the number of unmarked edges in the host graph $X$. The rule \texttt{next\_edge} is invoked at the beginning of \texttt{FORWARD}, and if it fails to match, the loop breaks. Clearly, an application of \texttt{next\_edge} reduces the measure $\#_1$ because the edge in the left-hand side is unmarked, and is marked red in the right-hand side. Since no rule in \texttt{FORWARD} creates or unmarks an edge, \texttt{FORWARD!} terminates. 
    
    We now show that the loop \texttt{DFS!} terminates. This time, define $\#_2X$ to be the number of non-blue edges in the host graph $X$. Suppose that $G \Rightarrow_{\texttt{FORWARD!}} H$. Since no rule in \texttt{FORWARD!} increases the number of non-blue edges, we have $\#_2G \geq \#_2H$.
    
    Let now $H \Rightarrow_{\texttt{back}} M$. Clearly, $\#_2H > \#_2M$ since the rule marks a dashed edge blue. As we know that \texttt{FORWARD!} is terminating, it follows that \texttt{DFS!} is also terminating.
\end{proof}
\begin{prop}
\label{prop:con-dashed}
    Throughout the execution of \texttt{DFS!}, (1) there is exactly one root node in the host graph, and (2) all dashed edges in the host graph form a directed path such that every node on the path is blue, and the endpoint of the path is the root node.
\end{prop}
\begin{proof}
    (1) Initially, every node of the host graph is non-rooted. Assuming that the input graph is non-empty, the main program applies the rule \texttt{init} to root exactly one node. Then, observe that every rule in the procedure \texttt{DFS} preserves the number of roots in the host graph. Therefore, (1) follows.
    
    (2) Upon inspection of the rules, \texttt{move} is the only rule dashing edges in the program. We proceed by induction on the execution steps of the program.

    Initially, before \texttt{move} is applied, no dashed edges exist in the host graph. The rule \texttt{init} creates a non-grey root node, and \texttt{back} cannot be applied. Thus, the invariant holds vacuously. Now, assume the invariant holds at some stage of execution. Clearly, the rule \texttt{move} extends the dashed path by adding one edge, and \texttt{back} reduces the dashed path by removing one edge. Both operations preserve the invariant and ensure the dashed path remains unique, its nodes remain blue, and the endpoint remains the root. Therefore, the invariant is maintained throughout the execution.
\end{proof}
\begin{prop}
\label{prop:every-node-non-grey-con}
    Let $G$ be a non-empty input graph to which \texttt{init} is applied. Then, the graph resulting from \texttt{DFS!} contains a connected component in which every node is non-grey.
\end{prop}
\begin{proof}
    The proposition is trivially true for an input graph consisting of a single node: \texttt{init} marks the node blue, and no rule marks a non-grey node grey. Assume now that the input graph contains at least two nodes. By Proposition \ref{prop:all-same-comp-con}, all non-grey nodes belong to the same connected component.\\ 
    
    \noindent \textit{Claim.} Each connected component containing at least one non-grey node does not contain any grey node. \\
    \noindent \textit{Proof of Claim.} For the sake of contradiction, suppose the connected component containing all non-grey nodes in the output graph also includes at least one grey node. Let $v$ and $u$ be such that $v$ is grey, $u$ is non-grey, and $v$ is adjacent to $u$, which must occur if they share a connected component.

    Upon inspection of the rules, \texttt{init} and \texttt{move} are the only rules altering the mark of a grey node; both mark a grey node blue and root it. Thus, $v$ must have been rooted at some point during the execution. Let us consider two cases based on whether $v$ remains rooted in the output graph.\\
    
    \textit{Case 1:}
        $v$ is rooted in the output graph. The procedure \texttt{DFS!} terminates only if \texttt{back} fails to apply, which occurs when \texttt{next\_edge} fails to find an unmarked edge. However, $u$, being non-grey, shares an unmarked edge with $v$. Since $u$ is blue and $v$ is rooted, the rules \texttt{next\_edge} and \texttt{move} must have applied, which contradicts the assumption that $v$ is grey in the output graph. \\
        
    \textit{Case 2:}
        $v$ is not rooted in the output graph. Throughout the execution of \texttt{DFS!}, by Proposition~\ref{prop:con-dashed}(1), there is exactly one root node at any time.
        Since $v$ is not rooted, it must have been unrooted by either \texttt{move}, which dashed an edge incident with $v$, or \texttt{back}. Two scenarios arise:
        \begin{enumerate}
            \item If \texttt{back} was applied to $v$, it would have unrooted $v$ only after no rules in \texttt{FORWARD} were applicable. However, since $v$ is adjacent to $u$, and $u$ shares an unmarked edge, \texttt{next\_edge} and \texttt{move} would have applied, turning $u$ non-grey. This contradicts the assumption.
            
            \item If \texttt{back} was not applied to $v$, $v$ must be part of a path of dashed edges with a root node as the endpoint (by Proposition \ref{prop:con-dashed}). The termination of \texttt{DFS!} implies that the rule \texttt{back} was no longer applicable, which is, again, a contradiction to the fact that there exists some root node adjacent to a blue node sharing a dashed edge (which follows from the fact that $v$ must be incident with at least one dashed edge).
    \end{enumerate}
    Given that both cases contradict the assumption, at the termination of \texttt{try init then DFS!} on a non-empty graph, there must be a connected component in which every node is non-grey.
\end{proof}
\begin{thm}[Correctness of \texttt{is-connected}]
    The program \texttt{is-connected} is totally correct\footnote{A program is \emph{totally correct} if it terminates on all input graphs and produces an output graph satisfying the program's specification.} with respect to the specification at the beginning of Subsection \ref{ss:problem-con}.
\end{thm}
\begin{proof}
    Observe that each rule can only change marks or rootedness, and therefore the resulting graph is isomorphic to the input graph up to marks.
    
    Termination follows from Lemma \ref{lem:term-con}. Let $G$ be the input graph. First, assume that the input graph $G$ is empty. Since the rule \texttt{init} is not applicable, the program immediately terminates and no \texttt{fail} command is invoked.

    Now, assume $G$ contains at least one node. Consider the graph $H$ as the output graph of \texttt{try init then DFS!} on $G$, i.e. $G \Rightarrow_{\texttt{init}} I \Rightarrow_{\texttt{DFS!}} H$ for some intermediate graph $I$. By Proposition \ref{prop:every-node-non-grey-con}, there exists some connected component in $H$ where every node is non-grey, let us call it $C$, and by Proposition \ref{prop:all-same-comp-con}, it follows that $H-C$ contains no non-grey node. Therefore, $G$ is not connected if and only if $H-C$ is non-empty. The procedure \texttt{Check} is executed after \texttt{DFS!}, and verifies that no grey node, which must belong to $H-C$ should one exist, exists. If a grey node exists, $H-C$ is non-empty and the program fails, indicating that $G$ is not connected. If none exists, the program terminates without invoking the \texttt{fail} command, indicating that $G$ is connected.
\end{proof}

\subsection{Proof of Complexity}

\begin{prop}
    \label{prop:red-edges-con}
    Throughout the execution of \texttt{is-connected} on an input graph, there is at most one red edge in the host graph.
\end{prop}
\begin{proof}
    Upon inspection of the rules, \texttt{next\_edge} is the only rule that marks an edge red in the host graph (note that the program is structure-preserving). It is easy to see that when \texttt{next\_edge} is applied, either \texttt{move} or \texttt{ignore} has to apply, as node \texttt{2} of \texttt{next\_edge} can either be blue or grey. Therefore, it follows that the number of red edges in the host graph is bounded to $1$.
\end{proof}

\begin{thm}[Complexity of \texttt{is-connected}]

The program \texttt{is-connected} terminates in time linear in the size of the input graph.
    \label{the:complexity-con}
\end{thm}
\begin{proof}
We now prove that the program \texttt{is-connected} terminates in time linear in the size of the input graph. The argument proceeds by bounding both the number of matching attempts and the number of rule applications. More specifically, we show (1) that for every rule, there is at most one matching attempt; then, (2) we derive an expression for the number of calls during the execution of the program for each rule; finally, (3) we combine those results to derive the time complexity of \texttt{is-connected}.

\begin{enumerate}

\item We consider the matching behaviour of each rule. The rule \texttt{init} attempts to match a single grey node. If the input graph is non-empty, this match succeeds immediately on the first node; otherwise, it fails without making further attempts. Hence, there is at most one matching attempt for \texttt{init}. The rule \texttt{match} also attempts to find a grey node using the compiler procedure \texttt{firstHostNode(m)}, which operates in constant time. As a result, \texttt{match} also makes only one matching attempt.

Since there is at most one root in the host graph (\texttt{init} is the only rule creating or deleting a root), node \texttt{1} of \texttt{next\_edge}, \texttt{move} and \texttt{ignore}, and node \texttt{2} of \texttt{back} match in constant time. The rule \texttt{next\_edge} matches any unmarked edge incident to the root and incident to an any-marked node. Since any node adjacent to the root is marked, and non-loop edges are stored in distinct lists with respect to their marks, there is at most one matching attempt for \texttt{next\_edge}. An analogous argument can be made for all rules beside \texttt{init} and \texttt{match}, as it is known from Proposition \ref{prop:red-edges-con} that there is at most one red edge in the host graph, and Proposition \ref{prop:con-dashed} establishes that a root node is incident to at most one dashed edge, hence limiting the number of possible matches to a single edge. The rule \texttt{match} requires the matching algorithm to find a grey node, which the procedure \texttt{firstHostNode(m)} of the updated compiler (Figure \ref{fig:complexity-assumptions}) executes in constant time.

\item We now analyse the number of times each rule is called during execution. For the purpose of this proof, let $n$ and $m$ be the number of nodes and edges, respectively. 

The rule \texttt{init} is invoked once at the beginning of execution and either succeeds (if the graph is non-empty) or fails. The rule \texttt{match} is also called exactly once in the procedure \texttt{Check}.

The rule \texttt{back} marks an edge blue and is called within the depth-first search loop. Since blue edges retain their mark, each edge can only be marked blue once. Thus, \texttt{back} may be successfully applied at most $m$ times. The loop containing \texttt{back} terminates when it fails to apply, so there is at most one unsuccessful application, giving us a total of at most $m+1$ calls.

The rules \texttt{move} and \texttt{ignore} are called only after a successful application of \texttt{next\_edge}. The rule \texttt{move} marks grey nodes blue, and since non-grey nodes do not revert to grey, at most $n-1$ applications are needed (i.e. one for each node except the one initially marked by \texttt{init}). The rule \texttt{ignore}, like \texttt{back}, marks edges blue, and thus can be applied at most $m$ times.

The rule \texttt{next\_edge} selects and marks an unmarked edge. Since edges cannot revert to being unmarked, it can be successfully applied at most $m$ times. Unsuccessful applications occur when no matching edge remains, which leads to a call to \texttt{back}. As previously noted, there are at most $m+1$ such calls, resulting in at most $m$ unsuccessful applications of \texttt{next\_edge}. Therefore, the total number of calls to \texttt{next\_edge} is bounded by $2m \in \mathrm{O}(m)$.

\item Since each rule application and match is performed in constant time, and the total number of applications across all rules is $\mathrm{O}(n + m)$, it follows that the overall time complexity of \texttt{is-connected} is linear in the size of the input graph (i.e. the number of nodes and edges).\qedhere 
\end{enumerate}
\end{proof}
\noindent 
Figure \ref{fig:is-con-bench} shows the empirical timings of \texttt{is-connected} on the graph classes of Figure \ref{fig:graph-class-1b} to Figure \ref{fig:graph-class-6b}. The measured runtimes do not account for graph parsing, building and printing, as these operations have a linear-time complexity with respect to the input size (Figure \ref{fig:complexity-assumptions}). Compilation time is also not included. Note that on discrete graphs, the program runs in constant time because no rule in \texttt{DFS} is applicable, and both \texttt{init} and \texttt{match} are applied only once.
\\

\begin{center}\begin{center}
    \resizebox{0.82\linewidth}{!}{\centering
\begin{minipage}{3.93cm}
\centering
\begin{tikzpicture}[scale=0.7]
	\node (a) at (-1.500,1.333)  [draw,circle,thick,fill=gray!50] {\,};
	\node (b) at (0.000,1.333)   [draw,circle,thick,fill=gray!50] {\,};
	\node (c) at (1.500,1.333)   [draw,circle,thick,fill=gray!50] {\,};
	\node (d) at (-1.500,0.000)  [draw,circle,thick,fill=gray!50] {\,};
	\node (e) at (0.000,0.000)   [draw,circle,thick,fill=gray!50] {\,};
	\node (f) at (1.500,0.000)   [draw,circle,thick,fill=gray!50] {\,};
	\node (g) at (-1.500,-1.333) [draw,circle,thick,fill=gray!50] {\,};
	\node (h) at (0.000,-1.333)  [draw,circle,thick,fill=gray!50] {\,};
	\node (i) at (1.500,-1.333)  [draw,circle,thick,fill=gray!50] {\,};
	
	\draw (a) edge[->, thick] (b)
	      (a) edge[->, thick] (d)
	      (b) edge[->, thick] (c)
	      (b) edge[->, thick] (e)
	      (c) edge[->, thick] (f)
	      (d) edge[->, thick] (e)
	      (d) edge[->, thick] (g)
	      (e) edge[->, thick] (f)
	      (e) edge[->, thick] (h)
	      (f) edge[->, thick] (i)
	      (g) edge[->, thick] (h)
	      (h) edge[->, thick] (i);
\end{tikzpicture}
\captionof{figure}{Square grid.}
\vspace{0.5cm}
\label{fig:graph-class-1b}
\end{minipage}
\begin{minipage}{3.93cm}
\centering
\begin{tikzpicture}[scale=0.7]
	\node (a) at (0.000,1.333)   [draw,circle,thick,fill=gray!50] {\,};
	\node (b) at (1.333,0.000)   [draw,circle,thick,fill=gray!50] {\,};
	\node (c) at (-1.333,0.000)  [draw,circle,thick,fill=gray!50] {\,};
	\node (d) at (2.000,-1.333)  [draw,circle,thick,fill=gray!50] {\,};
	\node (e) at (0.666,-1.333)  [draw,circle,thick,fill=gray!50] {\,};
	\node (f) at (-0.666,-1.333) [draw,circle,thick,fill=gray!50] {\,};
	\node (g) at (-2.000,-1.333) [draw,circle,thick,fill=gray!50] {\,};
	
	\draw (a) edge[->, thick] (b)
	      (a) edge[->, thick] (c)
	      (b) edge[->, thick] (d)
	      (b) edge[->, thick] (e)
	      (c) edge[->, thick] (f)
	      (c) edge[->, thick] (g);
\end{tikzpicture}
\captionof{figure}{Binary tree.}
\vspace{0.5cm}
\label{fig:graph-class-1c}
\end{minipage}
\begin{minipage}{3.93cm}
\centering
\begin{tikzpicture}[scale=0.7]
	\node (a) at (0.000,0.000)   [draw,circle,thick,fill=gray!50] {\,};
	\node (b) at (0.000,1.333)   [draw,circle,thick,fill=gray!50] {\,};
	\node (c) at (0.943,0.943)   [draw,circle,thick,fill=gray!50] {\,};
	\node (d) at (1.333,0.000)   [draw,circle,thick,fill=gray!50] {\,};
	\node (e) at (0.943,-0.943)  [draw,circle,thick,fill=gray!50] {\,};
	\node (f) at (0.000,-1.333)  [draw,circle,thick,fill=gray!50] {\,};
	\node (g) at (-0.943,-0.943) [draw,circle,thick,fill=gray!50] {\,};
	\node (h) at (-1.333,0.000)  [draw,circle,thick,fill=gray!50] {\,};
	\node (i) at (-0.943,0.943)  [draw,circle,thick,fill=gray!50] {\,};
	
	\draw (a) edge[->, thick] (b)
	      (c) edge[->, thick] (a)
	      (a) edge[->, thick] (d)
	      (e) edge[->, thick] (a)
	      (a) edge[->, thick] (f)
	      (g) edge[->, thick] (a)
	      (a) edge[->, thick] (h)
	      (i) edge[->, thick] (a);
\end{tikzpicture}
\captionof{figure}{Star graph.}
\vspace{0.5cm}
\label{fig:graph-class-2a}
\end{minipage}
\begin{minipage}{3.93cm}
\centering
\begin{tikzpicture}[scale=0.7]
	\node (a) at (0.0000,1.3333)   [draw,circle,thick,fill=gray!50] {\,};
	\node (b) at (1.1545,0.6666)   [draw,circle,thick,fill=gray!50] {\,};
	\node (c) at (1.1545,-0.6666)  [draw,circle,thick,fill=gray!50] {\,};
	\node (d) at (0.0000,-1.3334)  [draw,circle,thick,fill=gray!50] {\,};
	\node (e) at (-1.1545,-0.6666) [draw,circle,thick,fill=gray!50] {\,};
	\node (f) at (-1.1545,0.6666)  [draw,circle,thick,fill=gray!50] {\,};
	
	\draw (a) edge[->, thick] (b)
	      (b) edge[->, thick] (c)
	      (c) edge[->, thick] (d)
	      (d) edge[->, thick] (e)
	      (e) edge[->, thick] (f)
	      (f) edge[->, thick] (a);
\end{tikzpicture}
\captionof{figure}{Cycle graph.}
\vspace{0.5cm}
\label{fig:graph-class-2b}
\end{minipage}}
    \resizebox{0.82\linewidth}{!}{\begin{minipage}{6.9cm}
\centering
\begin{tikzpicture}[scale=0.7]
	\node (a) at (-1.500,1.333)  [draw,circle,thick,fill=gray!50] {\,};
	\node (b) at (1.500,1.333)   [draw,circle,thick,fill=gray!50] {\,};
	\node (c) at (-1.500,-1.333) [draw,circle,thick,fill=gray!50] {\,};
	\node (d) at (1.500,-1.333)  [draw,circle,thick,fill=gray!50] {\,};
	
	\draw 
            (a) edge[->, thick, loop left] (a)
	      (a) edge[->, thick, bend right=10=10] (b)
	      (a) edge[->, thick, bend right=10] (c)
	      (a) edge[->, thick, bend right=10] (d)
	      (b) edge[->, thick, bend right=10] (a)
	      (b) edge[->, thick, loop right] (b)
	      (b) edge[->, thick, bend right=10] (c)
	      (b) edge[->, thick, bend right=10] (d)
	      (c) edge[->, thick, bend right=10] (a)
	      (c) edge[->, thick, bend right=10] (b)
	      (c) edge[->, thick, loop left] (c)
	      (c) edge[->, thick, bend right=10] (d)
            (d) edge[->, thick, bend right=10] (a)
	      (d) edge[->, thick, bend right=10] (b)
	      (d) edge[->, thick, bend right=10] (c)
	      (d) edge[->, thick, loop right] (d);
\end{tikzpicture}
\captionof{figure}{Complete graph.}
\label{fig:graph-class-3a}
\end{minipage}
\begin{minipage}{4.88cm}
\centering
\begin{tikzpicture}[scale=0.7]
	\node (a) at (-2.5,0)  [draw,circle,thick,fill=gray!50] {\,};
	\node (b) at (-0.8333,0)   [draw,circle,thick,fill=gray!50] {\,};
	\node (c) at (0.8333,0) [draw,circle,thick,fill=gray!50] {\,};
	\node (d) at (2.500,0)  [draw,circle,thick,fill=gray!50] {\,};
        \node (X) at (0,1.333){\,};
        \node (y) at (0,-1.333){\,};
	
	\draw (a) edge[->, thick] (b)
	      (b) edge[->, thick] (c)
	      (c) edge[->, thick] (d);
\end{tikzpicture}
\captionof{figure}{Linked list.}
\label{fig:graph-class-3b}
\end{minipage}
\begin{minipage}{4.88cm}
    \centering
    \begin{tikzpicture}[scale=0.7]
        \node (b) at (-0.8333,0)   [draw,circle,thick,fill=gray!50] {\,};
        \node (c) at (0.8333,0) [draw,circle,thick,fill=gray!50] {\,};
        \node (d) at (2.500,0)  [draw,circle,thick,fill=gray!50] {\,};
            \node (X) at (0,1.333){\,};
            \node (y) at (0,-1.333){\,};

    \end{tikzpicture}
    \captionof{figure}{Discrete graph.}
    \label{fig:graph-class-6b}
\end{minipage}}
    \end{center}
\end{center}

\begin{center}
    \begin{figure}[!ht]
\centering
\begin{tikzpicture}
\begin{axis}[
  xlabel=size (number of nodes and edges),
  ylabel=runtime (ms), ylabel style={above=0.2mm},
  width=9.2cm,height=7.2cm,
  legend style={at={(1.65,1)}},
  ymajorgrids=true,
  grid style=dashed]
  \addplot[color=plot1, mark=square*] table [y=time, x=n]{Figures/Benchmarks/Connectedness/is-con-list.dat};
  \addlegendentry{List graphs}
  \addplot[color=plot2, mark=square*] table [y=time, x=n]{Figures/Benchmarks/Connectedness/is-con-cycle.dat};
  \addlegendentry{Cycle graphs}
  \addplot[color=red, mark=square*] table [y=time, x=n]{Figures/Benchmarks/Connectedness/is-con-grid.dat};
  \addlegendentry{Square grids}
  \addplot[color=plot4, mark=square*] table [y=time, x=n]{Figures/Benchmarks/Connectedness/is-con-tree.dat};
  \addlegendentry{Binary trees}
  \addplot[color=plot5, mark=square*] table [y=time, x=n]{Figures/Benchmarks/Connectedness/is-con-star.dat};
  \addlegendentry{{Star graphs}}
  \addplot[color=plot6, mark=square*] table [y=time, x=n]{Figures/Benchmarks/Connectedness/is-con-complete.dat};
  \addlegendentry{{Complete graphs}}
  \addplot[color=plot7, mark=square*] table [y=time, x=n]{Figures/Benchmarks/Connectedness/is-con-discrete.dat};
  \addlegendentry{{Discrete graphs}}
\end{axis}  
\end{tikzpicture}
\caption{Measured performance of the program \texttt{is-connected}.}
\label{fig:is-con-bench}
\end{figure}
\end{center}

\section{Case Study: Recognising Acyclic Graphs}
\label{s:acyclicity}
Checking whether a given graph contains a directed cycle is a basic problem in the area of graph algorithms \cite{Skiena20a}. A \gp{} program solving this problem is given in \cite{campbell2022fast}, but to run in linear time it requires input graphs of bounded node degree. The same paper contains a program for the related problem of recognising binary DAGs, which are acyclic graphs in which each node has at most two outgoing edges. This program has a linear runtime on arbitrary input graphs but is destructive in that the input graph is partially or totally deleted. 


\subsection{Program}
\label{ss:program-dag}

The program \texttt{is-dag}\footnote{The concrete syntax of the program available at: \url{https://gist.github.com/ismaili-ziad/959f26e188b210821d08eb5ed2965404}.} in Figure \ref{fig:is-dag-fig} recognises acylic graphs with respect to the following specification.

\begin{description}
\item[\textbf{Input}] An arbitrary \gp{} host graph such that
\begin{enumerate}
        \item each node is non-rooted and marked grey, and
        \item each edge is unmarked.
    \end{enumerate}
\item[\textbf{Output}] If the input graph is acyclic, a host graph that is isomorphic to the input graph up to marks. Otherwise \emph{failure}.
\end{description}

\definecolor{gp2pink}{RGB}{255, 153, 238}

\begin{figure}[!ht]
    \begin{mdframed}[linewidth=0.8pt]
    \begin{verbatim}
Main  = (init; DFS!; try unroot else break)!; Check
DFS   = try next_edge then (try {move, ignore} 
                            else (set_flag; break)) 
        else (try loop; try back else break)
Check = if flag then fail
    \end{verbatim}
    
    \begin{tikzpicture}
    \tikzstyle{every node}=[font=\ttfamily]
    
    \draw (0.9,0.75) node[align=left] {init(x:list)};
    \node[rectangle, thick, thick, rounded corners=7, draw=black, minimum size=5mm, fill=gray!50, label=below:\tiny\tiny 1](a2) at (0,0){x};
    \draw (1,0) node[] {$\Rightarrow$};
    \node[rectangle, thick, rounded corners=7, draw=black, minimum size=5mm,  double, double distance=1pt, fill=red!60, label=below:\tiny\tiny1](a2) at (2,0){x};
    \draw[->, thick];
    \draw (2.6,1) -- (2.6,-0.5);
    \end{tikzpicture}
    \begin{tikzpicture}
    \tikzstyle{every node}=[font=\ttfamily]
    \draw (1.1,0.75) node[align=left] {unroot(x:list)};
    \node[rectangle, thick, rounded corners=7, draw=black, minimum size=5mm,  double, double distance=1pt, fill=red!60, label=below:\tiny\tiny1](a2) at (0,0){x};
    \draw (1,0) node[] {$\Rightarrow$};
    \node[rectangle, thick, rounded corners=7, draw=black, minimum size=5mm, fill=cyan, label=below:\tiny\tiny 1](a2) at (2,0){x};
    \draw[->, thick];
    \draw (2.8,1) -- (2.8,-0.5);
    \end{tikzpicture}
    \begin{tikzpicture}
    \tikzstyle{every node}=[font=\ttfamily]
    \draw (1.3,0.75) node[align=left] {set\_flag(x:list)};
    \node[rectangle, thick, rounded corners=7, draw=black, minimum size=5mm,  double, double distance=1pt, fill=red!60, label=below:\tiny\tiny1](a2) at (0,0){x};
    \draw (1,0) node[] {$\Rightarrow$};
    \node[rectangle, thick, rounded corners=7, draw=black, minimum size=5mm, double, double distance=1pt, fill=green!60, label=below:\tiny\tiny 1](a2) at (2,0){x};
    \draw[->, thick];
    \draw (3.1,1) -- (3.1,-0.5);
    \end{tikzpicture}
    \begin{tikzpicture}
    \tikzstyle{every node}=[font=\ttfamily]
    \draw (0.9,0.75) node[align=left] {flag(x:list)};
    \node[rectangle, thick, rounded corners=7, draw=black, minimum size=5mm, double, double distance=1pt, fill=green!60, label=below:\tiny\tiny1](a2) at (0,0){x};
    \draw (1,0) node[] {$\Rightarrow$};
    \node[rectangle, thick, rounded corners=7, draw=black, minimum size=5mm, double, double distance=1pt, fill=green!60, label=below:\tiny\tiny 1](a2) at (2,0){x};
    \draw[->, thick];
    \end{tikzpicture}
    
    \begin{tikzpicture}
    \tikzstyle{every node}=[font=\ttfamily]
    \draw (1.8,0.75) node[align=left] {next\_edge(x,y,z:list)};
    \node[rectangle, thick, rounded corners=7, draw=black, minimum size=5mm,  double, double distance=1pt, fill=red!60, label=below:\tiny\tiny1](a1) at (0,0){x};
    \node[rectangle, thick, rounded corners=7, draw=black, minimum size=5mm, fill=gp2pink, label=below:\tiny\tiny2](a2) at (1,0){y};
    \draw (2,0) node[] {$\Rightarrow$};
    \node[rectangle, thick, rounded corners=7, draw=black, minimum size=5mm, double, double distance=1pt, fill=red!60, label=below:\tiny\tiny1](a3) at (3,0){x};
    \node[rectangle, thick, rounded corners=7, draw=black, minimum size=5mm, fill=gp2pink, label=below:\tiny\tiny2](a4) at (4,0){y};
    \draw[->, line width=1.2pt] (a1) edge[black] node[above, color = black]{z} (a2) (a3) edge[red] node[above, color = black]{z} (a4);
    \draw (5,1) -- (5,-0.5);
    \end{tikzpicture}
    \hspace{1em}
    \begin{tikzpicture}
    \tikzstyle{every node}=[font=\ttfamily]
    \draw (1.45,0.75) node[align=left] {ignore(x,y,z:list)};
    \node[rectangle, thick, rounded corners=7, draw=black, minimum size=5mm,  double, double distance=1pt, fill=red!60, label=below:\tiny\tiny1](a1) at (0,0){x};
    \node[rectangle, thick, rounded corners=7, draw=black, minimum size=5mm, fill=cyan, label=below:\tiny\tiny2](a2) at (1,0){y};
    \draw (2,0) node[] {$\Rightarrow$};
    \node[rectangle, thick, rounded corners=7, draw=black, minimum size=5mm, double, double distance=1pt, fill=red!60, label=below:\tiny\tiny1](a3) at (3,0){x};
    \node[rectangle, thick, rounded corners=7, draw=black, minimum size=5mm, fill=cyan, label=below:\tiny\tiny2](a4) at (4,0){y};
    \draw[->, line width=1.2pt] (a1) edge[red] node[above, color = black]{z} (a2) (a3) edge[cyan] node[above, color = black]{z} (a4);
    \end{tikzpicture}
    
    \begin{tikzpicture}
    \tikzstyle{every node}=[font=\ttfamily]
    \draw (1.3,0.75) node[align=left] {move(x,y,z:list)};
    \node[rectangle, thick, rounded corners=7, draw=black, minimum size=5mm,  double, double distance=1pt, fill=red!60, label=below:\tiny\tiny1](a1) at (0,0){x};
    \node[rectangle, thick, rounded corners=7, draw=black, minimum size=5mm, fill=gray!50, label=below:\tiny\tiny2](a2) at (1,0){y};
    \draw (2,0) node[] {$\Rightarrow$};
    \node[rectangle, thick, rounded corners=7, draw=black, minimum size=5mm, fill=red!60, label=below:\tiny\tiny1](a3) at (3,0){x};
    \node[rectangle, thick, rounded corners=7, draw=black, minimum size=5mm, double, double distance=1pt, fill=red!60, label=below:\tiny\tiny2](a4) at (4,0){y};
    \draw[->, line width=1.2pt] (a1) edge[red] node[above, color = black]{z} (a2) (a3) edge[dashed] node[above, color = black]{z} (a4);
    \draw (5,1) -- (5,-0.5);
    \end{tikzpicture}
    \hspace{1em}
    \begin{tikzpicture}
    \tikzstyle{every node}=[font=\ttfamily]
    \draw (1.3,0.75) node[align=left] {back(x,y,z:list)};
    \node[rectangle, thick, rounded corners=7, draw=black, minimum size=5mm, fill=red!60, label=below:\tiny\tiny1](a1) at (0,0){x};
    \node[rectangle, thick, rounded corners=7, draw=black, minimum size=5mm, double, double distance=1pt, fill=red!60, label=below:\tiny\tiny2](a2) at (1,0){y};
    \draw (2,0) node[] {$\Rightarrow$};
    \node[rectangle, thick, rounded corners=7, draw=black, minimum size=5mm, double, double distance=1pt, fill=red!60, label=below:\tiny\tiny1](a3) at (3,0){x};
    \node[rectangle, thick, rounded corners=7, draw=black, minimum size=5mm, fill=cyan, label=below:\tiny\tiny2](a4) at (4,0){y};
    \draw[->, line width=1.2pt] (a1) edge[dashed] node[above, color = black]{z} (a2) (a3) edge[cyan] node[above, color = black]{z} (a4);
    \end{tikzpicture}

    \begin{tikzpicture}
    \tikzstyle{every node}=[font=\ttfamily]
    \draw (1.1,0.65) node[align=left] {loop(x,z:list)};
    \node[rectangle, thick, rounded corners=7, draw=black, minimum size=5mm,  double, double distance=1pt, fill=red!60, label=below:\tiny\tiny1](a2) at (0,0){x};
    \draw (1,0) node[] {$\Rightarrow$};
    \node[rectangle, thick, rounded corners=7, draw=black, minimum size=5mm, double, double distance=1pt, fill=green!60, label=below:\tiny\tiny1](a3) at (2,0){x};
    \draw[->, line width=1.2pt] (a2) to [out=330,in=300,looseness=8] node[right, color = black]{z} (a2);
    \draw[->, line width=1.2pt] (a3) to [out=330,in=300,looseness=8] node[right, color = black]{z} (a3);
    \end{tikzpicture}
    
    \end{mdframed}
    \caption{The program \texttt{is-dag}.}
    \label{fig:is-dag-fig}
    \end{figure}

\noindent 
Figure \ref{fig:is-dag-ex} illustrates an execution of \texttt{is-dag} on a cyclic input graph, and Figure \ref{fig:is-dag-ex-acyclic} on an acylic graph. The program implements a directed DFS (depth-first search) of the host graph that marks the visited nodes red or blue, where the red nodes are currently being visited. All the red nodes are connected by a directed path of dashed edges, where the last node is rooted.

It is an invariant of \texttt{is-dag} that there is at most one root in the host graph throughout the program's execution (see Proposition~\ref{prop:unique-root}). The graph contains a cycle if and only if the search finds an edge from the root to a red node (see Proposition~\ref{prop:dag-dashed-red-green}).

\begin{figure}
    \centering
    \input{Figures/is-dag-ex}
    \caption{Sample execution of \texttt{is-dag} on an acyclic graph (\texttt{nx\_dg}, \texttt{mv}, \texttt{bk} and \texttt{unrt} correspond to \texttt{next\_edge}, \texttt{move}, \texttt{back} and \texttt{unroot}, respectively).}
    \label{fig:is-dag-ex}
\end{figure}

\begin{figure}
    \centering
    \usetikzlibrary{overlay-beamer-styles}

\definecolor{gp2green}{RGB}{153, 255, 170}
\definecolor{gp2blue}{RGB}{120, 161, 242}
\definecolor{gp2red}{RGB}{233, 73, 87}
\definecolor{gp2pink}{RGB}{255, 153, 238}
\definecolor{gp2grey}{RGB}{210, 210, 210}
\definecolor{plot1}{RGB}{70, 116, 193}
\definecolor{plot2}{RGB}{235, 125, 60}
\definecolor{plot3}{RGB}{165, 165, 165}
\definecolor{plot4}{RGB}{252, 190, 45}
\definecolor{plot5}{RGB}{94, 156, 210}
\definecolor{plot6}{RGB}{113, 171, 77}
\definecolor{plot7}{RGB}{156, 72, 25}
\definecolor{plot8}{RGB}{40, 69, 117}

\tikzset{gp2 node/.style={draw, circle, thick, minimum width=0.64cm}}
\tikzset{root node/.style={draw, circle, thick, minimum width=0.64cm, double, double distance=0.3mm}}

\scalebox{0.95}{%
\begin{tabular}{ccccccc}
\begin{minipage}{2.1cm}
\centering
\begin{tikzpicture}[scale=0.6]
	\node (a) at (-0.500,0.000)   [draw,circle,thick, fill=gray!50] {\,};
	\node (b) at (1.000,0.000)   [draw,circle,thick, fill=gray!50] {\,};
	\node (c) at (0.250,2.000)   [draw,circle,thick, fill=gray!50] {\,};
        \node (d) at (-1.000,2.000)   [draw,circle,thick, fill=gray!50] {\,};
	\draw (a) edge[->, thick] (b)
              (b) edge[->, thick] (c)
              (c) edge[->, thick] (a)
              (d) edge[->, thick] (a);
\end{tikzpicture}
\end{minipage}
&
$\Rightarrow_{\mtt{init}}$
&
\begin{minipage}{2.1cm}
\centering
\begin{tikzpicture}[scale=0.6]
	\node (a) at (-0.500,0.000)   [draw,circle,thick, fill=gray!50] {\,};
	\node (b) at (1.000,0.000)   [draw,circle,thick, fill=gray!50] {\,};
	\node (c) at (0.250,2.000)  [draw,circle,thick, fill=gray!50] {\,};
        \node (d) at (-1.000,2.000)   [root node,fill=red!60, inner sep=0pt, minimum size=0.4cm] {\,};
	\draw (a) edge[->, thick] (b)
              (b) edge[->, thick] (c)
              (c) edge[->, thick] (a)
              (d) edge[->, thick] (a);
\end{tikzpicture}
\end{minipage}
&
$\Rightarrow_{\mtt{nx\_dg}}$
&
\begin{minipage}{2.1cm}
\centering
\begin{tikzpicture}[scale=0.6]
	\node (a) at (-0.500,0.000)   [draw,circle,thick, fill=gray!50] {\,};
	\node (b) at (1.000,0.000)   [draw,circle,thick, fill=gray!50] {\,};
	\node (c) at (0.250,2.000)  [draw,circle,thick, fill=gray!50] {\,};
        \node (d) at (-1.000,2.000)   [root node,fill=red!60, inner sep=0pt, minimum size=0.4cm] {\,};
	\draw (a) edge[->, thick] (b)
              (b) edge[->, thick] (c)
              (c) edge[->, thick] (a)
              (d) edge[->, thick, red] (a);
\end{tikzpicture}
\end{minipage}
&
$\Rightarrow_{\mtt{mv}}$
&
\begin{minipage}{2.1cm}
\centering
\begin{tikzpicture}[scale=0.6]
	\node (a) at (-0.500,0.000)  [root node,fill=red!60, inner sep=0pt, minimum size=0.4cm] {\,};
	\node (b) at (1.000,0.000)   [draw,circle,thick, fill=gray!50] {\,};
	\node (c) at (0.250,2.000)  [draw,circle,thick, fill=gray!50] {\,};
        \node (d) at (-1.000,2.000)  [draw,circle,thick, fill=red!60] {\,};
	\draw (a) edge[->, thick] (b)
              (b) edge[->, thick] (c)
              (c) edge[->, thick] (a)
              (d) edge[->, thick, dashed] (a);
\end{tikzpicture}
\end{minipage}
\\
\\
 &&&&&& $\Downarrow_{\mtt{nx\_dg}}$
\\[1ex]
\begin{minipage}{2.1cm}
\centering
\begin{tikzpicture}[scale=0.6]
	\node (a) at (-0.500,0.000){};  	
    \node (a) at (-0.500,0.000)  [draw,circle,thick, fill=red!60] {\,};
	\node (b) at (1.000,0.000)   [draw,circle,thick, fill=red!60] {\,};
	\node (c) at (0.250,2.000)  [root node,fill=red!60, inner sep=0pt, minimum size=0.4cm] {\,};
        \node (d) at (-1.000,2.000)  [draw,circle,thick, fill=red!60] {\,};
	\draw (a) edge[->, thick, dashed] (b)
              (b) edge[->, thick, dashed] (c)
              (c) edge[->, thick] (a)
              (d) edge[->, thick, dashed] (a);
\end{tikzpicture}
\end{minipage}
&
$\Leftarrow_{\mtt{mv}}$
&
\begin{minipage}{2.1cm}
\centering
\begin{tikzpicture}[scale=0.6]
	\node (a) at (-0.500,0.000)  [draw,circle,thick, fill=red!60] {\,};
	\node (b) at (1.000,0.000) [root node,fill=red!60, inner sep=0pt, minimum size=0.4cm] {\,};
	\node (c) at (0.250,2.000)  [draw,circle,thick, fill=gray!50] {\,};
        \node (d) at (-1.000,2.000)  [draw,circle,thick, fill=red!60] {\,};
	\draw (a) edge[->, thick, dashed] (b)
              (b) edge[->, thick, red] (c)
              (c) edge[->, thick] (a)
              (d) edge[->, thick, dashed] (a);
\end{tikzpicture}
\end{minipage}
&
$\Leftarrow_{\mtt{nx\_dg}}$
&
\begin{minipage}{2.1cm}
\centering
\begin{tikzpicture}[scale=0.6]
	\node (a) at (-0.500,0.000)  [draw,circle,thick, fill=red!60] {\,};
	\node (b) at (1.000,0.000)   [root node,fill=red!60, inner sep=0pt, minimum size=0.4cm] {\,};
	\node (c) at (0.250,2.000)  [draw,circle,thick, fill=gray!50] {\,};
        \node (d) at (-1.000,2.000)  [draw,circle,thick, fill=red!60] {\,};
	\draw (a) edge[->, thick, dashed] (b)
              (b) edge[->, thick] (c)
              (c) edge[->, thick] (a)
              (d) edge[->, thick, dashed] (a);
\end{tikzpicture}
\end{minipage}
&
$\Leftarrow_{\mtt{mv}}$
&
\begin{minipage}{2.1cm}
\centering
\begin{tikzpicture}[scale=0.6]
	\node (a) at (-0.500,0.000)  [root node,fill=red!60, inner sep=0pt, minimum size=0.4cm] {\,};
	\node (b) at (1.000,0.000)   [draw,circle,thick, fill=gray!50] {\,};
	\node (c) at (0.250,2.000)  [draw,circle,thick, fill=gray!50] {\,};
        \node (d) at (-1.000,2.000)  [draw,circle,thick, fill=red!60] {\,};
	\draw (a) edge[->, thick, red] (b)
              (b) edge[->, thick] (c)
              (c) edge[->, thick] (a)
              (d) edge[->, thick, dashed] (a);
\end{tikzpicture}
\end{minipage}
\\
\\
$\Downarrow_{\mtt{nx\_dg}}$ &&&&&&
\\[1ex]
\begin{minipage}{2.1cm}
\centering
\begin{tikzpicture}[scale=0.6]
	\node (a) at (-0.500,0.000)  [draw,circle,thick, fill=red!60] {\,};
	\node (b) at (1.000,0.000)   [draw,circle,thick, fill=red!60] {\,};
	\node (c) at (0.250,2.000)  [root node,fill=red!60, inner sep=0pt, minimum size=0.4cm] {\,};
        \node (d) at (-1.000,2.000)  [draw,circle,thick, fill=red!60] {\,};
	\draw (a) edge[->, thick, dashed] (b)
              (b) edge[->, thick, dashed] (c)
              (c) edge[->, thick, red] (a)
              (d) edge[->, thick, dashed] (a);
\end{tikzpicture}
\end{minipage}
&
$\Rightarrow_{\mtt{set\_flag}}$
&
\begin{minipage}{2.1cm}
\centering
\begin{tikzpicture}[scale=0.6]
	\node (a) at (-0.500,0.000)  [draw,circle,thick, fill=red!60] {\,};
	\node (b) at (1.000,0.000)   [draw,circle,thick, fill=red!60] {\,};
	\node (c) at (0.250,2.000)  [root node,fill=green!60, inner sep=0pt, minimum size=0.4cm] {\,};
        \node (d) at (-1.000,2.000)  [draw,circle,thick, fill=red!60] {\,};
	\draw (a) edge[->, thick, dashed] (b)
              (b) edge[->, thick, dashed] (c)
              (c) edge[->, thick, red] (a)
              (d) edge[->, thick, dashed] (a);
\end{tikzpicture}
\end{minipage}
&
$\Rightarrow_{\mtt{flag}}$
&
\begin{minipage}{2.1cm}
\centering
\begin{tikzpicture}[scale=0.6]
	\node (a) at (-0.500,0.000)  [draw,circle,thick, fill=red!60] {\,};
	\node (b) at (1.000,0.000)   [draw,circle,thick, fill=red!60] {\,};
	\node (c) at (0.250,2.000)  [root node,fill=green!60, inner sep=0pt, minimum size=0.4cm] {\,};
        \node (d) at (-1.000,2.000)  [draw,circle,thick, fill=red!60] {\,};
	\draw (a) edge[->, thick, dashed] (b)
              (b) edge[->, thick, dashed] (c)
              (c) edge[->, thick, red] (a)
              (d) edge[->, thick, dashed] (a);
\end{tikzpicture}
\end{minipage}
&
$\Rightarrow_{\mtt{fail}}$
&
\begin{minipage}{2.1cm}
\centering
\begin{tikzpicture}[scale=0.6]
\tikzstyle{every node}=[]

\node[](v1) at (0.5,0.6){Failure};
\node[](v2) at (0,0){};
\end{tikzpicture}
\end{minipage}
\end{tabular}}
    \caption{Sample execution of \texttt{is-dag} on a cyclic graph.}
    \label{fig:is-dag-ex-acyclic}
\end{figure}

Consider the loop \texttt{(init; DFS!; try unroot else break)!} of \texttt{is-dag}'s main procedure. Rule \texttt{init} selects an arbitrary grey node as a root to start a directed DFS. The loop \texttt{DFS!} moves the root in depth-first fashion through the host graph. The procedure uses a \texttt{try-else} command to find any unprocessed (that is, unmarked) edge outgoing from the root. It does this by calling \texttt{next\_edge}. If there is such an edge, the rule marks it red so that it can be uniquely identified by the rest of the procedure. If no such edge exits, the root can no longer move forward and the \texttt{else} statement is invoked instead. 

After a successful application of \texttt{next\_edge}, the root is adjacent to either (1) a grey node, (2) a blue node, or (3) a red node. In case (1), the rule \texttt{move} moves the root to the grey node, marks it red and dashes the traversed edge. Dashed edges represent the path followed by the directed DFS. In case (2), the red edge is marked blue by the rule \texttt{ignore} so that it can no longer be matched by \texttt{next\_edge}. In case (3), neither \texttt{move} nor \texttt{ignore} is applicable so that \texttt{set\_flag} marks the root green, indicating the existence of a cycle. 

If \texttt{next\_edge} is not applicable, the command sequence \texttt{(try loop; try back else break)} is executed. Rule \texttt{loop} checks whether there is a loop attached to the root. If this is the case, the rule marks the root green, similar to \texttt{set\_flag}. Then rule \texttt{back} is tried which implements the \textit{pop} operation in the above mentioned stack model. The rule moves the root backwards along an incoming dashed edge. If no incoming dashed edge is present, the root must be the only element on the stack so that the \texttt{break} command terminates the loop \texttt{DFS!}. 

Upon termination of \texttt{DFS!}, the rule \texttt{unroot} attempts to turn the root into an unrooted blue node. If this is not possible, the root must have been marked green by \texttt{set\_flag} or \texttt{loop}. This implies the existence of a cycle and hence the outer loop of \texttt{is-dag} is terminated. 

If rule \texttt{unroot} could be applied, there may still be nodes that have not been visited by the DFS. These are nodes that are not directly reachable from the initial nodes chosen so far. In this case the execution of the outer loop is continued until \texttt{init} is no longer applicable or  \texttt{unroot} fails. 

\subsection{Proof of Correctness}

\begin{defn}[Input Graph]
    An \emph{input graph}, in the context of the acyclicity problem, is an arbitrary GP\,2 host graph such that: 
    \begin{enumerate}
        \item every node is marked grey and non-rooted, and
        \item every edge is unmarked.
    \end{enumerate}
    \label{def:input-dag}
\end{defn}
\begin{prop}
\label{prop:unique-root}
    Throughout the execution of \texttt{is-dag} on an input graph, there is at most one root node in the host graph.
\end{prop}
\begin{proof}
    Upon inspection of the rules, \texttt{init} and \texttt{unroot} are the only rules adding or deleting root nodes in the host graph, both invoked in the looping sequence \texttt{(init; DFS!; try unroot else break)!}.

    The input graph initially contains no root nodes according to the specification at the beginning of this section. On the first invocation of the parenthetical looping sequence in \texttt{Main}, \texttt{init} adds one root node to the host graph. The looping procedure \texttt{DFS!} does not add nor delete root nodes. The rule \texttt{unroot} deletes the node added by \texttt{init} should it apply, or exits the loop.
\end{proof}
\begin{prop}
\label{prop:path-red-dag}
    Let $G$ be a non-empty graph where all nodes are either grey or blue, at least one node is grey, all edges outgoing from blue nodes are marked blue, and all other edges are unmarked. Let $v$ be the grey node marked blue and rooted by the rule \texttt{init} on $G$. Throughout the execution of \texttt{DFS!} on $G$, all dashed edges in the host graph form a directed path such that all red nodes in the graph are in the path, $v$ is the starting node, a root node (marked either red or green) is the endpoint, and all nodes on the path, except possibly the last node, are red.
\end{prop}
\begin{proof}
    The invariant holds in $G$, i.e. before \texttt{DFS!} is invoked; here, the host graph contains exactly one red (rooted) node and no dashed edges. Now, assume the invariant holds at some stage of execution. Clearly, the rule \texttt{move} extends the dashed path by adding one dashed edge and one red node, and \texttt{back} reduces the dashed path by removing one edge. These two rules are the only ones creating dashed edges in the host graph, and no other rule deletes or adds root nodes (Proposition \ref{prop:unique-root}). Furthermore, \texttt{set\_flag} and \texttt{loop} are the only rules altering the mark of the root node, and mark it green. Therefore, the invariant is maintained throughout the execution. 
\end{proof}
\begin{prop}
\label{prop:dag-dashed-red-green}
    Let $G$ be a host graph and $v$ be a node as specified in Proposition \ref{prop:path-red-dag}. Upon the execution of \texttt{DFS!} on $G$, either \texttt{set\_flag} or \texttt{loop} is applied if and only if there exists some directed cycle all nodes of which are reachable from $v$.
\end{prop}
\begin{proof}
    As established in Proposition \ref{prop:path-red-dag}, there exists exactly one path of red nodes starting from $v$, where the endpoint is marked either green or red. Initially, this path consists of exactly one root node marked red, $v$.
    
    Observe that the rule \texttt{move} moves the root to a node directly reachable from the current root while leaving the previously rooted node marked red. The only rule that turns a red node blue is \texttt{back}, which corresponds to the \textit{pop} operation in the DFS. Once a node is turned blue, it will no longer be revisited (this can be verified by inspecting the marks of the rules). The nodes that remain red are precisely those ancestral to the current rooted endpoint in the DFS.
    
    If the rules \texttt{move} or \texttt{ignore} fail to apply, it implies that the second node (\texttt{node 2}) in these rules is already marked red. This condition indicates the presence of an edge from the current endpoint to a node already on the path, hence the existence of a directed cycle. Consequently, the rule \texttt{set\_flag} is applied, and the looping procedure \texttt{DFS!} terminates.
    
    Now, assume that some node $u$, reachable from $v$, has a looping edge. In this case, the rule \texttt{loop} applies before the invocation of \texttt{back}. The rule \texttt{loop} is called when the endpoint of the current path is $u$, which will eventually occur since all nodes reachable from $v$ will be rooted during the DFS traversal, provided there is no edge-cycle (i.e. a cycle that is not a looping edge).
\end{proof}
\begin{lem}[Termination of \texttt{is-dag}]
    The program \texttt{is-dag} terminates on any input graph.
    \label{lem:termination-dag}
\end{lem}
\begin{proof}
    Consider the following weight assignment to edge marks: $\texttt{unmarked} \mapsto 3$, $\texttt{red} \mapsto 2$, $\texttt{dashed} \mapsto 1$, and $\texttt{blue} \mapsto 0$.

    Given a host graph $X$, define $\#_1X$ to be the sum of all the edge weights, as defined above. Then, \texttt{DFS!} terminates because its loop body reduces $\#_1X$: rules \texttt{next\_edge}, \texttt{move}, \texttt{ignore}, and \texttt{back} all decrease the total weight, while no other rule increases it. 

    Now, define $\#_2X$ to be the number of grey nodes in $X$. Then, \texttt{init} decreases $\#_2X$ while \texttt{DFS!} and \texttt{unroot} do not increase the measure.

    Hence, \texttt{is-dag} terminates.
\end{proof}
\begin{thm}[Correctness of \texttt{is-dag}]
    \label{thm:correctness-dag}
    The program \texttt{is-dag} is totally correct with respect to the specification at the beginning of Subsection \ref{ss:program-dag}.
\end{thm}
\begin{proof}
    Termination follows from Lemma \ref{lem:termination-dag}. Let $G$ be the input graph. If $G$ is empty, \texttt{init} fails to match, \texttt{flag} fails to match, and the program terminates, outputting the empty graph and satisfying the specification. Suppose now that $G$ consists of at least one node. We then split the remainder of this proof into two cases and define the parenthetical sequence \texttt{(init; DFS!; try unroot else break)} as \texttt{R}.
    
        \textit{Case 1:} $G$ is acyclic. Observe that \texttt{R} terminates under two conditions: (1) \texttt{init} fails to apply, or (2) \texttt{unroot} fails to apply.

        The rule \texttt{init} fails to apply if and only if no grey node exists in the host graph, which indicates that all nodes have been visited in the DFS implemented by \texttt{DFS!}. By Proposition \ref{prop:dag-dashed-red-green}, the inapplicability of \texttt{init} implies that no directed cycle was found during any execution of \texttt{DFS!}, as \texttt{init} cannot be invoked after \texttt{unroot} fails to match. Consequently, \texttt{R} terminates, and the rule \texttt{flag} in \texttt{Check} fails to match. Furthermore, the rule \texttt{unroot} can never fail to apply in this context because the only rules capable of turning the root node green in \texttt{DFS!} are \texttt{set\_flag} and \texttt{loop}. Both of these rules apply if and only if $G$ contains a cycle (Proposition~\ref{prop:dag-dashed-red-green}), which it does not. Therefore, \texttt{unroot} will always succeed, leaving no green root for \texttt{flag} to apply successfully.

        \textit{Case 2:}
        $G$ is cyclic.
        Let $v$ be some node to which \texttt{init} is applied.
        
        \textit{Subcase 2.1:} A cycle is reachable from $v$. Then either the rule \texttt{set\_flag} or \texttt{loop} will eventually apply during the subsequent execution of \texttt{DFS!} (Proposition~\ref{prop:dag-dashed-red-green}). Once those rules are applied, \texttt{break} is immediately invoked, the rule \texttt{unroot} is not applicable, and therefore, \texttt{R} terminates. The rule \texttt{flag} will apply, causing the program to terminate with the execution of the \texttt{fail} command. 
        
        \textit{Subcase 2.2:} A cycle is not reachable from $v$. Then a cycle must exist within the set of nodes that are unreachable from $v$.\\
        \noindent \textit{Claim.} All the nodes on the cycle are grey. \\
        \noindent \textit{Proof of Claim.}
        Suppose by contradiction that there is at least one non-grey node on the cycle. Then, by Proposition~\ref{prop:dag-dashed-red-green}, at least one node that can reach this cycle had \texttt{init} applied to it. However, by Proposition~\ref{prop:dag-dashed-red-green}, this means that the program \texttt{is-dag} would have terminated with failure prior to the application of \texttt{init} on $v$, a contradiction. $\Box$
        
        By a straightforward induction on the number of remaining grey nodes, it follows that a node that can reach a cycle will eventually have \texttt{init} applied to it. By Proposition~\ref{prop:dag-dashed-red-green}, this implies that the program will eventually terminate by invocation of a \texttt{fail} command.
\end{proof}

\subsection{Proof of Complexity}

Let us examine the time complexity of \texttt{is-dag}.
\begin{prop}
    \label{prop:red-edges-dag}
    Throughout the execution of \texttt{is-dag}, there is at most one red edge in the host graph.
\end{prop}
\begin{proof}
    Upon inspection of the rules, \texttt{next\_edge} is the only rule that marks an edge red in the host graph (again, note that the program is structure-preserving). It is easy to see that when \texttt{next\_edge} is applied, either \texttt{move} or \texttt{ignore} has to apply, or both fail to apply, triggering the looping sequence to terminate. Therefore, it follows that the number of red edges in the host graph is bounded to $1$.
\end{proof}

\begin{thm}[Complexity of \texttt{is-dag}]
The program \texttt{is-dag} terminates in time linear in the size of the input graph.
    \label{the:complexity-dag}
\end{thm}
\begin{proof}
    We first show that, for every rule, there is at most one matching attempt with respect to the complexity assumptions of the updated compiler (Figure \ref{fig:complexity-assumptions}).

    An analogous argument to the proof of Theorem \ref{the:complexity-con} can be made with respect to the rules \texttt{init}, \texttt{next\_edge} (Proposition \ref{prop:red-edges-dag}), \texttt{ignore}, \texttt{move}, and \texttt{back}. Since there is at most one root node in the host graph, the rules \texttt{flag}, \texttt{loop} and \texttt{set\_flag} only require at most one matching attempt. 
    
    Similarly, analogous arguments can be made to prove that the rules \texttt{next\_edge}, \texttt{ignore} and \texttt{move} are called at most a number of times linear to the size of the graph. It is easy to see that the rule \texttt{init} can be called at most $n$ times, where $n$ is the number of nodes in the input graph (which does not vary). The rule \texttt{set\_flag} is only called once as its application results in the termination of \texttt{(init; DFS!; try unroot else break)!}. The rule \texttt{loop} precedes the invocation of \texttt{back} and is therefore called as many times as the latter, thus a number of times linear to the size of the input graph. Finally, the rule \texttt{flag} is only called once in \texttt{Check}. Therefore, the program \texttt{is-dag} runs in time linear to the size of the input graph.
\end{proof}

\begin{center}
    \begin{figure}[!ht]
    \centering
    \begin{tikzpicture}
    \begin{axis}[
      xlabel=size (number of nodes and edges),
      ylabel=runtime (ms), ylabel style={above=0.2mm},
      width=9.2cm,height=7.2cm,
      legend style={at={(1.58,,0.64)}},
      ymajorgrids=true,
      grid style=dashed]
      \addplot[color=plot2, mark=square*] table [y=time, x=n]{Figures/Benchmarks/DAG/is-dag-list.dat};
      \addlegendentry{List graphs}
      \addplot[color=plot1, mark=square*] table [y=time, x=n]{Figures/Benchmarks/DAG/is-dag-cycle.dat};
      \addlegendentry{Cycle graphs}
      \addplot[color=plot4, mark=square*] table [y=time, x=n]{Figures/Benchmarks/DAG/is-dag-grid.dat};
      \addlegendentry{Square grids}
      \addplot[color=plot5, mark=square*] table [y=time, x=n]{Figures/Benchmarks/DAG/is-dag-tree.dat};
      \addlegendentry{Binary trees}
      \addplot[color=plot7, mark=square*] table [y=time, x=n]{Figures/Benchmarks/DAG/is-dag-discrete.dat};
      \addlegendentry{Discrete graphs}
            \addplot[color=plot6, mark=square*] table [y=time, x=n]{Figures/Benchmarks/DAG/is-dag-star.dat};
      \addlegendentry{{Star graphs}}
    \end{axis}  
    \end{tikzpicture}
    \caption{Measured performance of the program \texttt{is-dag}.}
    \label{fig:bench-is-dag}
\end{figure}
\end{center}

Empirical runtimes are shown in Figure \ref{fig:bench-is-dag} for various graph classes of bounded degree (Figures \ref{fig:graph-class-1b}, \ref{fig:graph-class-1c}, \ref{fig:graph-class-2b}, \ref{fig:graph-class-3b} and \ref{fig:graph-class-6b}) and unbounded degree (Figure \ref{fig:graph-class-2a}).

\section{Case Study: The Bellman-Ford Algorithm}
\label{s:bf}
In this section, we present the first implementation of a single-source shortest-path graph algorithm in GP\,2 with a time complexity comparable to that of conventional imperative programming languages: $\mathrm{O}(nm)$, where $n$ and $m$ denote the numbers of nodes and edges, respectively. It computes the shortest distances from one node to all other nodes in the graph. The Bellman-Ford algorithm has been chosen for its simplistic elegance and its detection of negative-weight cycles, in contrast to other popular alternatives, such as Dijkstra's shortest-path algorithm, which only accepts graphs without negative-weight cycles.

We will first examine a pseudo-code representation of the Bellman-Ford algorithm.

\begin{algorithm}
\caption{Classic Bellman-Ford Algorithm}
\begin{algorithmic}[1]
  \STATE Input: A weighted digraph $G=(V,E)$ of $n$ nodes and a source node $s \in V$
  \STATE Set \texttt{dist}$[v] = \infty$ for all $v$ in $V$
  \STATE Set \texttt{dist}$[s] = 0$
  \FOR{$n-1$ iterations}
    \FOR{each edge $e$ in $E$}
      \STATE \texttt{dist}$[t(e)] = min(\texttt{dist}[s(e)]+w(e),\ \texttt{dist}[t(e)])$
    \ENDFOR
  \ENDFOR
  \FOR{each edge $e$ in $E$}
    \IF{$\texttt{dist}[s(e)]+w(e) < \texttt{dist}[t(e)]$}
        \STATE Return \texttt{FALSE}
    \ENDIF
  \ENDFOR
  \STATE Return $\texttt{dist}[\cdot]$
\end{algorithmic}
\label{a:bellman}
\end{algorithm}

Algorithm~\ref{a:bellman}, adapted from \cite{Sedgewick02a}, provides a high-level overview of the classical Bellman-Ford algorithm. The chief idea is to iteratively decrease the shortest distances \texttt{dist}$[\cdot]$ in the graph until the threshold of $n-1$ iterations is reached, where $n$ is the number of nodes in the graph. The following gives a descriptive account of each section of the code.\\

\textbf{Lines 1 to 3.} The algorithm begins by initialising an array \texttt{dist}$[\cdot]$. \texttt{dist}$[v]$ defines the lowest possible accumulated cost to go to $v$ from the source node $s$ in $G$. The upper bound distance for each node is initialised to $\infty$. Naturally, the distance from the source $s$ to itself is $0$.

\textbf{Lines 4 to 8.} This stage is called the \emph{relaxation phase}. At each iteration, the algorithm iterates through all the edges in $E$ and updates their upper bound whenever it finds a shorter distance. Let us consider the edge $e = (u,\ v)$. At the $i\textsuperscript{th}$ iteration ($1 \le i \le n-1$), the algorithm checks whether \texttt{dist}$[v]$ from the $(i-1)\textsuperscript{th}$ iteration (we can consider the initialisation phase to be the $0\textsuperscript{th}$ iteration) can be further improved if there exists a shorter path to $v$ via $u$ with a weighted cost of $\texttt{dist}[u]+w(e)$. That process is repeated $n-1$ times.

\textbf{Lines 9 to 13.} This section of the code detects negative cycles. The previous phase guarantees the computation of all shortest paths from $s$ to all other nodes $v$ in $G$, given that no negative cycle exists. If a negative cycle exists, it follows that there always exists a more optimal solution. The algorithm iterates through all edges and verifies whether a relaxation can still decrease at least one upper bound (which should have all been previously minimised). Should that be possible, this implies that a negative cycle exists and hence triggers the code to return \texttt{FALSE}.

\textbf{Line 14.} If the algorithm did not terminate during the negative cycle detection, it then successfully returns \texttt{TRUE} and exits.

\subsection{Program}
\label{ss:program-bf}

The program \texttt{bellman-ford}\footnote{The concrete syntax of the program is available at: \url{https://gist.github.com/ismaili-ziad/adc54b05aaff9ef6b3c57e7ea8a68543}.} (Figure \ref{fig:bf-fig}) expects an input graph whose edges are labelled with integers. It calculates the shortest distances from a unique root node to all nodes reachable via directed paths. For each such node, the program appends the shortest distance from the root (an integer value) to the node's label. The exact specification is as follows:

\begin{description}
\item[\textbf{Input}] An arbitrary \gp{} host graph such that
\begin{enumerate}
        \item every node is marked grey,
        \item there is exactly one root node,
        \item every edge is unmarked and labelled with an integer, and
        \item there are no looping edges.\footnote{We make this assumption for simplicity. It is straightforward to extend the program to handle looping edges.}
    \end{enumerate}
\item[\textbf{Output}] If the input graph contains a directed cycle whose sum of edge labels is negative and there is a directed path from the root to that cycle, \emph{failure}.\\
Otherwise, a graph obtained from the input graph as follows: 
\begin{enumerate}
        \item for each node reachable from the root, append its shortest distance to the label, 
        \item append the string \texttt{f} to the label of all other nodes, and 
        \item mark all edges blue.
    \end{enumerate}
\end{description}

\definecolor{gp2pink}{RGB}{255, 153, 238}

\begin{figure}[!ht]
    \begin{mdframed}[linewidth=1pt]
    \begin{verbatim}
Main  = set_counter; count!; (decrement; Relax!; Clean!)!; Final
Relax = root1; try no_deg 
               else ((unmarked_edge; try {unvisited, reduce}; 
                    finish)!; unroot1)
Clean = root2; unmark_edge!; unroot2
Final = (root1; (unmarked_edge; if reduce 
                                then set_flag; finish)!; 
                                unroot1)!;
        if flag then fail; delete_counter; no_deg_inv!
    \end{verbatim}
    
    \begin{tikzpicture}
    \tikzstyle{every node}=[font=\ttfamily]
    
    \draw (1.6,0.75) node[align=left] {set\_counter(x:list)};
    \node[rectangle, thick, rounded corners=7, draw=black, minimum size=5mm,  double, double distance=1pt, fill=gray!50 , label=below:\tiny\tiny1](a2) at (0,0){x};
    \draw (1,0) node[] {$\Rightarrow$};
    \node[rectangle, thick, rounded corners=7, draw=black, minimum size=5mm, fill=cyan, label=below:\tiny\tiny 1](a2) at (2,0){x:0};
    \node[rectangle, thick, rounded corners=7, draw=black, minimum size=5mm, fill=green!60, label=below:\tiny\tiny ](a3) at (3.33,0){0};
    \draw[->, thick] (a3) edge[dashed] (a2);
    \draw (4.33,1) -- (4.33,-0.5);
    \end{tikzpicture}
    \hspace{1em}
    \begin{tikzpicture}
    \tikzstyle{every node}=[font=\ttfamily]
    \draw (1.6,0.75) node[align=left] {count(x:list;i:int)};
    \node[rectangle, thick, rounded corners=7, draw=black, minimum size=5mm, fill=gray!50, label=below:\tiny\tiny 1](a1) at (0,0){x};
    \node[rectangle, thick, rounded corners=7, draw=black, minimum size=5mm, fill=green!60, label=below:\tiny\tiny 2](a2) at (1,0){i};
    \draw (2,0) node[] {$\Rightarrow$};
    \node[rectangle, thick, rounded corners=7, draw=black, minimum size=5mm, fill=cyan, label=below:\tiny\tiny 1](a3) at (3.5,0){x:"f"};
    \node[rectangle, thick, rounded corners=7, draw=black, minimum size=5mm, fill=green!60, label=below:\tiny\tiny 2](a4) at (5,0){i+1};
    \draw[-, thick];
    \end{tikzpicture}
    \\
    \begin{tikzpicture}
    \tikzstyle{every node}=[font=\ttfamily]
    \draw (1,0.75) node[align=left] {root1(x:list)};
    \node[rectangle, thick, rounded corners=7, draw=black, minimum size=5mm, fill=cyan, label=below:\tiny\tiny 1](a1) at (0,0){x};
    \draw (1,0) node[] {$\Rightarrow$};
    \node[rectangle, thick, rounded corners=7, draw=black, minimum size=5mm,  double, double distance=1pt, fill=cyan, label=below:\tiny\tiny1](a2) at (2,0){x};
    \draw[-, thick];
    \draw (3,1) -- (3,-0.5);
    \end{tikzpicture}
    \hspace{1em}
    \begin{tikzpicture}
    \tikzstyle{every node}=[font=\ttfamily]
    \draw (1.1,0.75) node[align=left] {no\_deg(x:list)};
    \node[rectangle, thick, rounded corners=7, draw=black, minimum size=5mm,  double, double distance=1pt, fill=cyan, label=below:\tiny\tiny\,](a1) at (0,0){x};
    \draw (1,0) node[] {$\Rightarrow$};
    \node[rectangle, thick, rounded corners=7, draw=black, minimum size=5mm, label=below:\tiny\tiny\,](a2) at (2,0){x};
    \draw[-, thick];
    \draw (0,-0.62) -- (0,-0.62);
    \draw (3,1) -- (3,-0.5);
    \end{tikzpicture}
    \hspace{1em}
    \begin{tikzpicture}
    \tikzstyle{every node}=[font=\ttfamily]
    \draw (1.15,0.75) node[align=left] {unroot1(x:list)};
    \node[rectangle, thick, rounded corners=7, draw=black, minimum size=5mm,  double, double distance=1pt, fill=cyan, label=below:\tiny\tiny1](a1) at (0,0){x};
    \draw (1,0) node[] {$\Rightarrow$};
    \node[rectangle, thick, rounded corners=7, draw=black, minimum size=5mm, fill=gray!50, label=below:\tiny\tiny 1](a2) at (2,0){x};
    \draw[-, thick];
    \end{tikzpicture}
    \\
    \begin{tikzpicture}
    \tikzstyle{every node}=[font=\ttfamily]
    \draw (1.38,0.75) node[align=left] {decrement(i:int)};
    \node[rectangle, thick, rounded corners=7, draw=black, minimum size=5mm, fill=green!60, label=below:\tiny\tiny 1](a1) at (0,0){i};
    \draw (1,0) node[] {$\Rightarrow$};
    \node[rectangle, thick, rounded corners=7, draw=black, minimum size=5mm, fill=green!60, label=below:\tiny\tiny1](a2) at (2.33,0){i-1};
    \draw[-, thick];
    \draw (3.66,1) -- (3.66,-0.9);
    \draw (0.8,-0.85) node[align=left] {where i > 0};
    \end{tikzpicture}
    \hspace{1em}
    \begin{tikzpicture}
    \tikzstyle{every node}=[font=\ttfamily]
    \draw (2.2,0.75) node[align=left] {unmarked\_edge(x,y,z:list)};
    \node[rectangle, thick, rounded corners=7, draw=black, minimum size=5mm,  double, double distance=1pt, fill=cyan, label=below:\tiny\tiny1](a1) at (0,0){x};
    \node[rectangle, thick, rounded corners=7, draw=black, minimum size=5mm, fill=gp2pink, label=below:\tiny\tiny2](a2) at (2,0){y};
    \draw (3,0) node[] {$\Rightarrow$};
    \node[rectangle, thick, rounded corners=7, draw=black, minimum size=5mm,  double, double distance=1pt, fill=cyan, label=below:\tiny\tiny1](a3) at (4,0){x};
    \node[rectangle, thick, rounded corners=7, draw=black, minimum size=5mm, fill=gp2pink, label=below:\tiny\tiny2](a4) at (6,0){y};
    \draw[->, line width=1.2pt] (a1) edge node[above, color = black]{z} (a2) (a3) edge[red] node[above, color = black]{z} (a4);
    \draw (0,-1.1) -- (0,-1.1);
    \end{tikzpicture}
    \\
    \begin{tikzpicture}
    \tikzstyle{every node}=[font=\ttfamily]
    \draw (2,0.75) node[align=left] {unmark\_edge(x,y,z:list)};
    \node[rectangle, thick, rounded corners=7, draw=black, minimum size=5mm,  double, double distance=1pt, fill=cyan, label=below:\tiny\tiny1](a1) at (0,0){x};
    \node[rectangle, thick, rounded corners=7, draw=black, minimum size=5mm, fill=gp2pink, label=below:\tiny\tiny2](a2) at (2,0){y};
    \draw (3,0) node[] {$\Rightarrow$};
    \node[rectangle, thick, rounded corners=7, draw=black, minimum size=5mm,  double, double distance=1pt, fill=cyan, label=below:\tiny\tiny1](a3) at (4,0){x};
    \node[rectangle, thick, rounded corners=7, draw=black, minimum size=5mm, fill=gp2pink, label=below:\tiny\tiny2](a4) at (6,0){y};
    \draw[->, line width=1.2pt] (a1) edge[cyan] node[above, color = black]{z} (a2) (a3) edge node[above, color = black]{z} (a4);
    \draw (6.95,1) -- (6.95,-0.5);
    \end{tikzpicture}
    \hspace{1em}
    \begin{tikzpicture}
    \tikzstyle{every node}=[font=\ttfamily]
    \draw (1.5,0.75) node[align=left] {no\_deg\_inv(x:list)};
    \node[rectangle, thick, rounded corners=7, draw=black, minimum size=5mm, label=below:\tiny\tiny1](a1) at (0,0){x};
    \draw (1,0) node[] {$\Rightarrow$};
    \node[rectangle, thick, rounded corners=7, draw=black, minimum size=5mm, fill=gray!50, label=below:\tiny\tiny 1](a2) at (2,0){x};
    \draw[-, thick];
    \end{tikzpicture}
    \\
    \begin{tikzpicture}
    \tikzstyle{every node}=[font=\ttfamily]
    \draw (2.25,0.75) node[align=left] {unvisited(x,y:list;s,w:int)};
    \node[rectangle, thick, rounded corners=7, draw=black, minimum size=5mm,  double, double distance=1pt, fill=cyan, label=below:\tiny\tiny1](a1) at (0,0){x:s};
    \node[rectangle, thick, rounded corners=7, draw=black, minimum size=5mm, fill=gp2pink, label=below:\tiny\tiny2](a2) at (2,0){y:"f"};
    \draw (3.1,0) node[] {$\Rightarrow$};
    \node[rectangle, thick, rounded corners=7, draw=black, minimum size=5mm,  double, double distance=1pt, fill=cyan, label=below:\tiny\tiny1](a3) at (4,0){x:s};
    \node[rectangle, thick, rounded corners=7, draw=black, minimum size=5mm, fill=gp2pink, label=below:\tiny\tiny2](a4) at (6,0){y:(s+w)};
    \draw[->, line width=1.2pt] (a1) edge[red] node[above, color = black]{w} (a2) (a3) edge[red] node[above, color = black]{w} (a4);
    \draw (7.33,1) -- (7.33,-0.5);
    \end{tikzpicture}
    \hspace{1em}
    \begin{tikzpicture}
    \tikzstyle{every node}=[font=\ttfamily]
    \draw (1,0.75) node[align=left] {root2(x:list)};
    \node[rectangle, thick, rounded corners=7, draw=black, minimum size=5mm, fill=gray!50, label=below:\tiny\tiny 1](a1) at (0,0){x};
    \draw (1,0) node[] {$\Rightarrow$};
    \node[rectangle, thick, rounded corners=7, draw=black, minimum size=5mm, double, double distance=1pt, fill=cyan, label=below:\tiny\tiny1](a2) at (2,0){x};
    \draw[-, thick];
    \end{tikzpicture}
    \\
    \begin{tikzpicture}
    \tikzstyle{every node}=[font=\ttfamily]
    \draw (2.15,0.75) node[align=left] {reduce(x,y:list;s,t,w:int)};
    \node[rectangle, thick, rounded corners=7, draw=black, minimum size=5mm,  double, double distance=1pt, fill=cyan, label=below:\tiny\tiny1](a1) at (0,0){x:s};
    \node[rectangle, thick, rounded corners=7, draw=black, minimum size=5mm, fill=gp2pink, label=below:\tiny\tiny2](a2) at (2,0){y:t};
    \draw (3.1,0) node[] {$\Rightarrow$};
    \node[rectangle, thick, rounded corners=7, draw=black, minimum size=5mm,  double, double distance=1pt, fill=cyan, label=below:\tiny\tiny1](a3) at (4,0){x:s};
    \node[rectangle, thick, rounded corners=7, draw=black, minimum size=5mm, fill=gp2pink, label=below:\tiny\tiny2](a4) at (6,0){y:(s+w)};
    \draw[->, line width=1.2pt] (a1) edge[red] node[above, color = black]{w} (a2) (a3) edge[red] node[above, color = black]{w} (a4);
    \draw (7.33,1) -- (7.33,-0.9);
    \draw (0.9,-0.85) node[align=left] {where s+w < t};
    \end{tikzpicture}
    \hspace{1em}
    \begin{tikzpicture}
    \tikzstyle{every node}=[font=\ttfamily]
    \draw (1.15,0.75) node[align=left] {unroot2(x:list)};
    \node[rectangle, thick, rounded corners=7, draw=black, minimum size=5mm,  double, double distance=1pt, fill=cyan, label=below:\tiny\tiny1](a1) at (0,0){x};
    \draw (1,0) node[] {$\Rightarrow$};
    \node[rectangle, thick, rounded corners=7, draw=black, minimum size=5mm, fill=cyan, label=below:\tiny\tiny 1](a2) at (2,0){x};
    \draw[-, thick];
    \draw (0,-1.1) -- (0,-1.1);
    \end{tikzpicture}
    \\
    \begin{tikzpicture}
    \tikzstyle{every node}=[font=\ttfamily]
    \draw (1.5,0.75) node[align=left] {finish(x,y,z:list)};
    \node[rectangle, thick, rounded corners=7, draw=black, minimum size=5mm,  double, double distance=1pt, fill=cyan, label=below:\tiny\tiny1](a1) at (0,0){x};
    \node[rectangle, thick, rounded corners=7, draw=black, minimum size=5mm, fill=gp2pink, label=below:\tiny\tiny2](a2) at (2,0){y};
    \draw (3.1,0) node[] {$\Rightarrow$};
    \node[rectangle, thick, rounded corners=7, draw=black, minimum size=5mm,  double, double distance=1pt, fill=cyan, label=below:\tiny\tiny1](a3) at (4,0){x};
    \node[rectangle, thick, rounded corners=7, draw=black, minimum size=5mm, fill=gp2pink, label=below:\tiny\tiny2](a4) at (6,0){y};
    \draw[->, line width=1.2pt] (a1) edge[red] node[above, color = black]{z} (a2) (a3) edge[cyan] node[above, color = black]{z} (a4);
    \draw (7,1) -- (7,-0.5);
    \end{tikzpicture}
    \hspace{1em}
    \begin{tikzpicture}
    \tikzstyle{every node}=[font=\ttfamily]
    \draw (1.3,0.75) node[align=left] {set\_flag(x:list)};
    \node[rectangle, thick, rounded corners=7, draw=black, minimum size=5mm, fill=green!60, label=below:\tiny\tiny1](a1) at (0,0){x};
    \draw (1,0) node[] {$\Rightarrow$};
    \node[rectangle, thick, rounded corners=7, draw=black, minimum size=5mm, fill=green!60, label=below:\tiny\tiny 1](a2) at (2,0){-1};
    \draw[-, thick];
    \end{tikzpicture}
    \\
    \begin{tikzpicture}
    \tikzstyle{every node}=[font=\ttfamily]
    \draw (2.05,0.75) node[align=left] {delete\_counter(x,y:list)};
    \node[rectangle, thick, rounded corners=7, draw=black, minimum size=5mm,  double, double distance=1pt, fill=cyan, label=below:\tiny\tiny1](a2) at (3.33,0){x};
    \draw (2.33,0) node[] {$\Rightarrow$};
    \node[rectangle, thick, rounded corners=7, draw=black, minimum size=5mm, fill=cyan, label=below:\tiny\tiny 1](a2) at (0,0){x};
    \node[rectangle, thick, rounded corners=7, draw=black, minimum size=5mm, fill=green!60, label=below:\tiny\tiny ](a3) at (1.33,0){y};
    \draw[->, thick] (a3) edge[dashed] (a2);
    \draw (4.8,1) -- (4.8,-0.5);
    \end{tikzpicture}
    \hspace{0.8em}
    \begin{tikzpicture}
    \tikzstyle{every node}=[font=\ttfamily]
    \draw (0.25,0.75) node[align=left] {flag()};
    \node[rectangle, thick, rounded corners=7, draw=black, minimum size=5mm, fill=green!60, label=below:\tiny\tiny1](a1) at (0,0){-1};
    \draw (1,0) node[] {$\Rightarrow$};
    \node[rectangle, thick, rounded corners=7, draw=black, minimum size=5mm, fill=green!60, label=below:\tiny\tiny 1](a2) at (2,0){-1};
    \draw[-, thick];
    \end{tikzpicture}
    
    \end{mdframed}
    \caption{The program \texttt{bellman-ford}.}
    \label{fig:bf-fig}
    \end{figure}

Note that the input graph may contain negative-weight cycles reachable from the root, as the program is designed to detect them. The program first creates a \emph{counter} node labelled with the integer $n-1$, where $n$ is the number of nodes in the input graph. The node is marked green, a mark unique to that node to allow for fast, constant-time access. The rule \texttt{set\_counter} initialises this process by creating a green node labelled $0$ and marking the root node blue to indicate it has been counted. The rule also unroots the node, creates a dashed edge from the counter to store the location of the original root, and appends the value $0$ to the previously rooted node's label to denote that its distance to itself is $0$.

The rule \texttt{count} is then applied iteratively to increment the counter for each grey node in the graph. As nodes are counted, they are marked blue to prevent duplicate counting. Additionally, the rule appends the string \texttt{f} to each node's label. Here, \texttt{f} serves as a representation of infinity, in practice commonly represented by a large positive integer in Bellman-Ford implementations, but is implemented as a string to avoid dependency on the integer size limitations of the GP\,2 compiler. Define $l(v)$ as the value appended to the list label of node $v$ by \texttt{set\_counter} or \texttt{count}.

The sequence \texttt{(decrement; Relax!; Clean!)} implements the relaxation phase of the Bellman-Ford algorithm. The rule \texttt{decrement} decrements the value in the counter by $1$ should that value be non-zero; this is to ensure that the relaxation phase happens exactly $n-1$ times. Initially, all nodes incident to at least one edge are marked blue in the host graph, no root exists, and all edges are unmarked. A single execution of \texttt{Relax} accomplishes the following:

\begin{enumerate}
    \item The rule \texttt{root1} selects a blue node and roots it. At this stage, there is exactly one root in the host graph.
    \item The rule \texttt{no\_deg} checks whether the root node is isolated (notice that the node is not in the interface). If it is, it unmarks and unroots it to prevent the procedure \texttt{Relax} from considering it again in a future iteration. Isolated nodes do not have outgoing edges, so they need not be considered in the relaxation phase.
    \item In case the node is not isolated, the sequence \texttt{(unmarked\_edge; try \{unvisited, reduce\})} is invoked as long as possible. \texttt{unmarked\_edge} selects an unmarked edge outgoing from the root node and marks it red; let $u$ and $v$ be the source and target of the edge, respectively, and $w$ the weight of the edge; then two possible cases arise:
    \begin{enumerate}
        \item $l(u)$ is an integer value (i.e. not the string \texttt{f}). If $l(v) = \texttt{f}$, then the rule \texttt{unvisited} applies and sets $l(v)$ to $l(u)+w$. Otherwise, if $l(u)+w$ is strictly lesser than $l(v)$, the rule \texttt{reduce} applies and sets $l(v)$ to $l(u)+w$; else, none of the rules in the rule set \texttt{\{unvisited, reduce\}} applies.
        \item $l(u)$ is the string \texttt{f}. None of the rules in the rule set \texttt{\{unvisited, reduce\}} applies.
    \end{enumerate}
    \item The rule \texttt{finish} marks the selected edge blue to denote that it has been processed, so that it is no longer considered for relaxation until all other unmarked edges have been relaxed.
    \item The rule \texttt{unroot1} unroots the selected root node and unmarks it, so that it is no longer considered as a source node for relaxation until all other grey nodes have been considered. After the application of this rule, there is no root node in the host graph.
\end{enumerate}
Following an execution of \texttt{Relax!}, all edges in the host graph are marked blue (due to \texttt{finish}), and all non-isolated nodes (i.e. nodes whose degrees are greater than $0$) are marked grey (due to \texttt{unroot1}). The looping procedure \texttt{Clean!} prepares the host graph for another round of edge relaxations by unmarking all edges and turning all grey nodes blue.

The procedure \texttt{Final} checks that no negative-weight cycle reachable from the source exists in the host graph by running a modified version of the procedure \texttt{Relax!} once more. Note that the counter node is labelled $0$ at this stage. Should the rule \texttt{reduce} be applicable, the counter node's label is set to $-1$ by the rule \texttt{set\_flag}.

Finally, the program checks the applicablity of the rule \texttt{flag}, which can only apply if and only if \texttt{set\_flag} had been applied. If \texttt{flag} is applicable, it entails that a negative cycle was found, and the program fails. Otherwise, no negative cycle exists in the host graph and \texttt{no\_deg\_inv!} marks all unmarked nodes (i.e. isolated nodes) grey, as they were in the input graph.

\subsection{Proof of Correctness}

We now examine the correctness of the program \texttt{bellman-ford}.

From now on, we call graphs satisfying the assumptions of the specification given at the beginning of the subsection as \emph{input graphs}.
\begin{prop}
    \label{prop:counter}
Let $G$, $I$ and $H$ be graphs such that $G$ is an input graph, and $G \Rightarrow_{\texttt{set\_counter}} I \Rightarrow_{\texttt{count!}} H$. Then $H$ is obtained from $G$ as follows:

\begin{enumerate}
    \item All nodes are marked blue.
    \item The integer $0$ is appended to the list label of the root node.
    \item The string \texttt{f} is appended to the list labels of other nodes (i.e. non-root nodes).
    \item A new green-marked node is added, labelled with the integer $n-1$, where $n$ is the number of nodes in $G$ (prior to the addition of the node).
    \item A new dashed-marked edge is added, with the green-marked node as the source and the root node in $G$ as the target.
    \item The root node is unrooted.
\end{enumerate}
\end{prop}
\begin{proof}
    Let us first look at the application of the rule \texttt{set\_counter}; that is, $G\Rightarrow_{\texttt{set\_counter}}I$. Initially, all nodes in $G$ are grey and one node is rooted. Thus, the application of the rule \texttt{set\_counter} marks the root node blue, appends $0$ to its list label, creates a green node labelled $0$ and generates a dashed edge from that newly created green node pointing at the previously rooted node. Now, let us examine the application of \texttt{count!}; that is, $I \Rightarrow_{\texttt{count!}} H$. $I$ consists of $n+1$ nodes and $m$ edges, where $n$ and $m$ are the numbers of nodes and edges in $G$, with one green node, one blue node and $n-1$ grey nodes. The rule \texttt{count} marks one grey node blue, appends the string \texttt{f} to its list label, and increments the counter by $1$. Since the rule is applied as many times as there are grey nodes in $I$, the counter is labelled $n-1$ in $H$. Hence, $H$ consists of one green node labelled $0$, $n$ blue nodes, $m$ unmarked edges retaining their source and target nodes (notice that none of the rules involves edges), and one dashed edge from the green node to the node rooted in $G$. In addition, the node rooted in $G$ has $0$ appended to its list label in $H$, and all other nodes have the string \texttt{f} appended to their list labels.
\end{proof}
\begin{lem}[Termination of \texttt{Relax!}]
\label{lem:term-relax}
    The looping procedure \texttt{Relax!} terminates.
\end{lem}
\begin{proof}
    Let us show that the loop \texttt{(unmarked\_edge; try {unvisited, reduce}; finish)!} terminates first. To show termination, consider the measure $\#X$ consisting of the number of unmarked edges in the host graph $X$. The rule \texttt{unmarked\_edge} is invoked at the beginning of the loop, and if it fails to match, the loop breaks. Clearly, an application of \texttt{unmarked\_edge} reduces the measure $\#$. Since the number of edges is finite and no rule in the loop unmarks an edge, the rule eventually fails and the loop terminates.

    Coming back to \texttt{Relax!}, now consider $\#X$ to be the measure consisting of the number of blue nodes in the host graph $X$. Clearly, an application of \texttt{root1} reduces the measure $\#$. Since there is a finite number of nodes in the graph and no other rule in the procedure marks a node blue, the rule eventually fails and \texttt{Relax!} terminates.
\end{proof}
\begin{lem}[Termination of \texttt{Clean!}]
\label{lem:term-clean}
    The looping procedure \texttt{Clean!} terminates.
\end{lem}
\begin{proof}
    Define $\#X$ as the measure consisting of the number of grey nodes in the host graph $X$. An application of \texttt{root2} clearly reduces that measure, and since there is a finite number of nodes in the host graph and no rule in \texttt{Clean} marks a node grey, the rule eventually fails and the loop terminates.
\end{proof}
\begin{prop}
    \label{prop:n-1-times}
    Upon the execution of the program \texttt{bellman-ford} on an input graph $G$, the looping sequence \texttt{(decrement; Relax!; Clean!)!} terminates after exactly $n-1$ iterations, where $n$ is the number of nodes in $G$.
\end{prop}
\begin{proof}
    Termination of \texttt{Relax!} and \texttt{Count!} follow from Lemmata \ref{lem:term-relax} and \ref{lem:term-clean}.
    
    Recall from Proposition \ref{prop:counter} that the host graph contains exactly one green node labelled $n-1$ after the execution of the loop \texttt{count!}. The looping sequence terminates if and only if the rule \texttt{decrement} fails to apply. However, since the rule \texttt{decrement} is the only rule in the sequence to involve the counter node, then the sequence must be invoked exactly $n-1$ times before \texttt{decrement} can no longer apply and the sequence terminates.
\end{proof}
\begin{prop}
    \label{prop:relax-root1}
    Upon termination of \texttt{Relax!}, all non-counter and non-isolated nodes in the host graph are grey.
\end{prop}
\begin{proof}
    We first show that if \texttt{Relax!} has terminated, then the last call to \texttt{root1} failed.
    
    For the sake of contradiction, suppose that the last call to \texttt{root1} succeeded. We show that \texttt{root1} is called again, contradicting our assumption.

    Upon inspection of the program sequence, a successful application of \texttt{root1} leads to a call to \texttt{no\_deg}. If the latter succeeds, then \texttt{Relax} will be called, implying that \texttt{root1} will be called again. If \texttt{no\_deg} fails to apply, then \texttt{(unmarked\_edge; try {unvisited, reduce}; finish)!} is executed. If this loop terminates, then \texttt{unroot1} will be applied because the root node is blue, implying that \texttt{Relax} and hence \texttt{root1} will be applied again. If the loop does not terminate, then \texttt{Relax!} does not terminate, contradicting Lemma~\ref{lem:term-relax}. 
    
    Since \texttt{Relax!} terminates, the last call to \texttt{root1} failed. This implies that there is no blue node in the host graph. It is easy to see that unmarked nodes can only be produced by \texttt{no\_deg}, which applies to isolated nodes only. Hence, there are neither blue nor unmarked nodes in the graph, and thus all non-counter and non-isolated nodes must be grey.
\end{proof}

\begin{prop}
    \label{prop:relax}
    On a host graph following the application of \texttt{decrement} during the execution of \texttt{bellman-ford} on an input graph, the looping procedure \texttt{Relax!} relaxes each edge in the host graph exactly once.
\end{prop}
\begin{proof}
    Initially, all non-isolated nodes are blue; recall that isolated nodes need not be considered in the relaxation phase as they do not admit of any edge incident to them.
    
    Once a blue node is rooted by \texttt{root1}, either \texttt{no\_deg} applies and \texttt{Relax!} moves on to its next iteration, or it does not apply and the looping sequence \texttt{(unmarked\_edge; try \{unvisited, reduce\}; finish)!} is executed instead. Given that all edges outgoing from non-counter nodes are unmarked initially, the rule \texttt{unmarked\_edge} applies for any edge outgoing from the root node following the application of \texttt{root1} on that node. After such a rule is applied, if the root node's distance written on its list label is infinity (i.e. the string \texttt{f}), then the edge cannot be used to reduce the distance inscribed on the target of the edge and none of the rules in the rule set \texttt{\{unvisited, reduce\}} applies, and the program subsequently applies \texttt{finish} to mark the edge blue, to denote that it has been processed in the current iteration of \texttt{Relax!}. Otherwise, the last variable on the list label of the root is an integer; in that case, either the edge is used to reduce the distance of the target node via the application of either \texttt{unvisited} in case the distance of the target node is still infinity, or the application of \texttt{reduce} should the weight added to the distance of the root be strictly smaller than the distance inscribed on the target of the edge, or none of those applies (the edge cannot be used to improve the distance on the target node) and the program applies \texttt{finish} to mark the edge blue to denote it has been processed in the iteration of \texttt{Relax!}.

    The loop \texttt{(unmarked\_edge; try \{unvisited, reduce\}; finish)!} terminates if and only if \texttt{unmarked\_edge} is inapplicable (observe that \texttt{finish} will always apply following an application of \texttt{unmarked\_edge}). Therefore, the loop will terminate when all outgoing edges are processed and turned blue by \texttt{finish}. Since, by Proposition~\ref{prop:relax-root1}, the loop \texttt{Relax!} terminates once all blue nodes are processed by \texttt{root1}, all unmarked edges are incident to some blue node in the host graph and \texttt{(unmarked\_edge; try \{unvisited, reduce\}; finish)!} terminates once all unmarked outgoing edges are processed by \texttt{unmarked\_edge} and relaxed, then the looping procedure \texttt{Relax!} relaxes all edges in the host graph.

    Finally, observe that no rule in the procedure \texttt{Relax} changes the mark of a blue edge; therefore, every edge is relaxed exactly once.
\end{proof}
The next proposition establishes a graph-theoretical property that is needed in the subsequent proof.

\begin{prop}[Path Size]
\label{lem:size-path}
    Let $G$ be an input graph where $s$ the root node in $G$, no negative-weight cycle is reachable from $s$, and $t$ any node reachable from $s$ in $G$. Let $P$ be the shortest path of edges from $s$ to $t$ such that the sum of the weights of the edges in the path is minimal among all paths from $s$ to $t$. Then, the number of edges in $P$ is at most $n-1$, where $n$ is the number of nodes in $G$.
\end{prop}
\begin{proof}
    For the sake of contradiction, assume that $P$ contains more than $n-1$ edges. Since $G$ contains $n$ nodes, this implies that $P$ must revisit at least one node, forming a cycle in the path. Let this cycle be denoted as $C$.

    If the total weight of the edges in $C$ is positive, removing the cycle would result in a path with a smaller total weight, contradicting the assumption that the sum of the weights of the edges in $P$ is minimal among all paths from $s$ to $t$. If the total weight of the edges in $C$ is zero, removing the cycle does not change the total weight of the path. In this case, we can remove $C$ and obtain a shorter path (in terms of the number of edges) with the same weight, contradicting the minimality of $P$ as a shortest path.

    In all cases, we arrive at a contradiction. Therefore, $P$ cannot contain more than $n-1$ edges. Since the number of edges in any simple path (i.e. a path that does not contain cycles) in $G$ is at most $n-1$, the shortest path $P$ from $s$ to $t$ must also satisfy this condition.
\end{proof}
For the remainder of the proof, define $l(v)$ as the value appended to the list label of node $v$ by \texttt{set\_counter} or \texttt{count} in the program \texttt{bellman-ford}.

\begin{prop}[Computation of Shortest Distances]
\label{lem:shortest-dist}
     After $k$ iterations of the loop \texttt{(decrement; Relax!; Clean!)!}, for every non-counter node $v$ in the host graph reachable from the root node $r$ in the input graph by a path of at most $k$ edges, $l(v)$ is set to the shortest weighted distance from $r$ to $v$ by a path of at most $k$ edges.
\end{prop}
\begin{proof}
    Let us proceed by induction on the number of iterations of the looping sequence \texttt{(decrement; Relax!; Clean!)!}.

    When $k=0$ (base case), the program set $l(r)$ to $0$, and for all other nodes $v \neq r$, $l(v)$ was assigned the value $\texttt{f}$, representing infinity. At this stage, no edges were considered by the procedure \texttt{Relax}, so the shortest path to $r$ from $r$, trivially of length $0$, is correctly set to $0$, and the shortest path to all other nodes remains undefined (i.e. infinity, the string \texttt{f}). Thus, the lemma holds for $k=0$.

    Now, assume the property holds for some $k$; we show it still holds for $k+1$. Consider the $(k+1)\textsuperscript{th}$ iteration of \texttt{(decrement; Relax!; Clean!)!}. During this iteration, the procedure \texttt{Relax!} processes all edges in the host graph (Proposition \ref{prop:relax}). Indeed, for each edge $(u, v)$ with weight $w$, the program checks whether the distance to $v$ can be improved by relaxing the edge $(u, v)$ from a node $u$ whose current shortest distance $l(u)$ has already been computed in prior iterations (by the induction hypothesis).
    The relaxation phase \texttt{Relax!} updates\footnote{$\min\big(l(v), l(u) + w\big)=l(u)+w$ if $l(v)$ is not an integer.} $l(v)$ as follows:
    $$
    l(v) \gets \min\big(l(v), l(u) + w\big).
    $$
    If there exists a path from $r$ to $v$ consisting of at most $k+1$ edges, it must either:
    \begin{enumerate}
        \item pass through a node $u$ such that the shortest path to $u$ uses at most $k$ edges, followed by the edge $(u, v)$; or
        \item be a path already computed during the first $k$ iterations, using at most $k$ edges.
    \end{enumerate}
    Therefore, during the $(k+1)\textsuperscript{th}$ iteration, the looping procedure \texttt{Relax!} correctly updates $l(v)$ to the shortest distance from $r$ to $v$ using at most $k+1$ edges.
\end{proof}
\begin{prop}
    \label{prop:flag-bf}
    Upon the execution of \texttt{bellman-ford} on an input graph $G$, the rule \texttt{set\_flag} is applied if and only if $G$ contains a negative-weight cycle reachable from the root node.
\end{prop}
\begin{proof}
    Let $r$ be the root node in $G$. From Proposition \ref{lem:size-path}, it is established that the shortest path of edges from $r$ to some arbitrary non-counter node $v$ in the host graph such that the sum of the weights of the edges in the path is minimal among all paths from $r$ to $v$ consists of at most $n-1$ edges (where $n$ is the number of nodes in $G$), provided no negative-weight cycle reachable from $r$ exists. Proposition \ref{lem:shortest-dist} states that $k$ iterations of the loop \texttt{(decrement; Relax!; Clean!)!}, called prior to \texttt{Final}, compute the smallest weighted distances from $r$ to any non-counter node $v$ in the host graph by a path of at most $k$ edges. Furthermore, Proposition \ref{prop:n-1-times} establishes that this loop terminates after $n-1$ iterations. Therefore, upon termination of \texttt{(decrement; Relax!; Clean!)!}, the program has computed the smallest distances from $r$ to all other non-counter nodes in the host graph by a path of at most $n-1$ edges.

    The procedure \texttt{Final} executes one iteration of \texttt{Relax}, during which only the rule \texttt{reduce} is considered for relaxation. If the rule \texttt{reduce} is applicable during this iteration, it implies that there exists some node $u$ with $l(u)$ set to the shortest distance from $r$ to $u$ (by a path of at most $n-1$ edges), and an edge $(u, v)$ with weight $w$ such that $l(u) + w < l(v)$. In this case, the rule \texttt{reduce} would update $l(v)$ to $l(u) + w$.

    For such a relaxation to occur after $n-1$ iterations, there must exist a path from $r$ to $v$ consisting of more than $n-1$ edges that has a smaller total weight than any path of at most $n-1$ edges. This is only possible if there exists a negative-weight cycle reachable from $r$. Indeed, by repeatedly traversing the negative-weight cycle, the total path weight can decrease indefinitely, which would for a smaller value of $l(v)$ to be computed.

    When the rule \texttt{reduce} applies, the rule \texttt{set\_flag} subsequently sets the counter node's label to $-1$. This serves as an indicator that a negative-weight cycle was detected. Finally, conversely, if no negative-weight cycle is reachable from $r$, then no edge relaxation can occur during the \texttt{Final} procedure, and the rule \texttt{set\_flag} is never applied.
\end{proof}
\begin{thm}[Correctness of \texttt{bellman-ford}]
    \label{thm:correctness-bf}
    The program \texttt{bellman-ford} is totally correct with respect to the specification at the beginning of Subsection \ref{ss:program-bf}.
\end{thm}
\begin{proof}
    Termination easily follows from Lemmata \ref{lem:term-relax} and \ref{lem:term-clean}, and Proposition \ref{prop:n-1-times}. An argument analogous to the proof of Lemma \ref{lem:term-relax} can be made with respect to the termination of \texttt{Final}. 
    
    Let $G$ be the input graph. Let us split the remainder of this proof into two cases.
    
    \textit{Case 1:}
        $G$ contains a negative-weight cycle reachable from the source. Upon the termination of \texttt{count!}, the counter node is labelled $n-1$, where $n$ is the number of nodes in $G$. By Proposition \ref{prop:n-1-times}, the rule \texttt{decrement} is applied exactly $n-1$. Therefore, upon termination of the loop \texttt{(decrement; Relax!; Clean!)!}, the counter node is labelled $0$. As per Proposition \ref{prop:flag-bf}, \texttt{set\_flag} will apply during the execution of \texttt{Final}, and given that \texttt{set\_flag} is the only rule tempering with the list label of the counter node, an application of \texttt{flag} would entail that \texttt{set\_flag} was applied, and a negative-weight cycle reachable from the root was found. The program would subsequently invoke the \texttt{fail} command and fail, satisfying the specifications of the program. 

    \textit{Case 2:}
        $G$ does not contain a negative-weight cycle reachable from the source. Analogously to the previous case, the counter node is labelled $n-1$ upon termination of the loop \texttt{(decrement; Relax!; Clean!)!}. Given that no negative-weight cycle reachable from the source node exists in $G$, the rule \texttt{set\_flag} is never applied and the \texttt{fail} command is never invoked. The rule \texttt{delete\_counter} returns the root to the source node and deletes the counter (observe that no rule beside \texttt{set\_counter} and \texttt{delete\_counter} involves a dashed edge). Finally, as all shortest distances from the source can be achieved by a path of at most $n-1$ edges (Proposition \ref{lem:size-path}), the loop \texttt{(decrement; Relax!; Clean!)!} terminates after $n-1$ iterations (Proposition \ref{prop:n-1-times}), and $k$ iterations of the loop label the shortest weighted distances from the root to all other nodes by a path of at most $k$ edges (Lemma \ref{lem:shortest-dist}), all nodes reachable from the source have an integer appended to their list label denoting the smallest weighted distance from the root. Nodes that are not reachable from the source retained the string \texttt{f}.
\end{proof}

\subsection{Proof of Complexity}

Let us now analyse the time complexity of the program \texttt{bellman-ford}.

\begin{prop}
    \label{prop:constant-bf}
    Every rule in the program \texttt{bellman-ford} matches or fails to match in constant time with respect to the complexity assumptions of the updated compiler (Figure \ref{fig:complexity-assumptions}).
\end{prop}
\begin{proof}
    Since the input graph contains exactly one root node, a call to the rule \texttt{set\_counter} requires only a constant-time lookup in the set of root nodes.

    The rule \texttt{count} matches node \texttt{1} in constant time since the list label of the node can be arbitrary, and node \texttt{2} in constant time as well since there exists at most one green node in the host graph. Node \texttt{1} of \texttt{decrement}, \texttt{set\_flag} and \texttt{flag} is also matched in constant time due to the bound on the number of green nodes.

    The rule \texttt{root1} matches any blue node in the host graph of arbitrary list label, and thus in constant time. An analogous argument can be made for \texttt{root2}. Since there is exactly one root node after an application of \texttt{root1}, \texttt{no\_deg} and \texttt{unroot1} match in constant time. An analogous argument can be made for \texttt{unroot2} with respect to \texttt{root2}.

    The rule \texttt{unmarked\_edge} matches node \texttt{1} in constant time due to the bound on the number of roots in the host graph, and matches any unmarked edge of arbitrary list label outgoing from the root in constant time. The rules \texttt{finish}, \texttt{unvisited} and \texttt{reduce} match node \texttt{1} and the red edge in constant time, since there is at most one red edge throughout the execution of the program \texttt{bellman-ford}. The rule \texttt{unmark\_edge} matches node \texttt{1} and the blue edge in constant time, given that the latter can have an arbitrary list label.

    Finally, \texttt{delete\_counter} matches node \texttt{1} and the dashed edge in constant time, as there is at most one dashed edge throughout the execution of the program. And the rule \texttt{no\_deg\_inv} matches in constant time given that the list label of the node can be arbitrary.
\end{proof}
The following proposition is a simple graph-theoretic property that we subsequently need.

\begin{prop}
    \label{prop:connected-bf-2m}
    For every graph $G$, the number of non-isolated nodes is at most $2|E_G|$.
\end{prop}

\begin{proof}
    Given a node $v$, we write $\deg(v)$ for its degree.\footnote{The degree of a node is the sum of its indegree and outdegree.} 
    It is a basic fact that every graph $G$ satisfies $\sum_{v \in V_G} \deg(v) = 2|E_G|$ (see, for example, \cite{Bang_Jensen-Gutin09a}). Clearly, if $v$ is not isolated, we have $\deg(v) \geq 1$.  Thus there can be at most $2|E_G|$ non-isolated nodes.
\end{proof}

\begin{thm}[Complexity of \texttt{bellman-ford}]
    On any class of input graphs, the program \texttt{bellman-ford} terminates in time $\mathrm{O}(nm)$, where $n$ and $m$ are the numbers of nodes and edges in the input graph, respectively.
    \label{the:complexity-bf}
\end{thm}
\begin{proof}
    We established in Proposition \ref{prop:constant-bf} that, for every rule, there is at most one matching attempt with respect to the complexity assumptions of the updated compiler (Figure \ref{fig:complexity-assumptions}). Let us now examine how many times each rule is called.

    Clearly, the rules \texttt{set\_counter}, \texttt{delete\_counter} and \texttt{flag} are called at most once. Following an application of \texttt{set\_counter}, there are $n-1$ grey nodes in the host graph, where $n$ is the number of nodes in the input graph. The loop terminates when there is no grey node left, thus it is called $n$ times. Proposition \ref{prop:n-1-times} states that the looping sequence \texttt{(decrement; Relax!; Clean!)} terminates after $n-1$ iterations, thus the rule \texttt{decrement} is called $n$ times. Let that sequence be referred to as \texttt{R}. Prior to the very first iteration of \texttt{R}, during the execution of the program, isolated nodes are not unmarked. Following the iteration, every isolated node is unmarked. Therefore, the rule \texttt{root1} is called $n+1$ times in the first iteration of \texttt{R} in the execution of the program, and $n'+1$ times on each iteration afterward, where $n'$ is the number of non-isolated non-counter nodes in the host graph. The rule \texttt{no\_deg} is called after each application of \texttt{root1}, thus $n$ times on the first iteration of \texttt{R}, and $n'$ times afterward. The loop \texttt{(unmarked\_edge; try {unvisited, reduce}; finish)!} is called as many times as there are unmarked edges in the host graph and non-isolated nodes, thus $m+n'$ times where $m$ is the number of edges in the input graph. The rule \texttt{unroot1} is called $n'$ times, that is, as many times as there are non-isolated nodes.
    
    In each execution of \texttt{Clean!}, the rule \texttt{root2} is called $n'+1$ times, and the rule \texttt{unroot2}, $n'$ times. The rule \texttt{unmark\_edge} is called $m+n'$ times ($m$ successful applications and $n'$ failed attempts to terminate the loop on each non-isolated node). Finally, the rule \texttt{root1} in \texttt{Final} is called at most $n'+1$ times, \texttt{unmarked\_edge} is called at most $m+n'$ times, \texttt{reduce} at most $m$ times and \texttt{finish}, at most $m$ times.

    By Proposition \ref{prop:connected-bf-2m}, $n'$ is at most $2m$. Therefore, on all iterations of \texttt{{(decrement; Relax!; Clean!)!}} but the first one, a number of calls at most linear to the number of edges are called. On the first iteration, a number of calls linear to the size of the graph are called.

    Therefore, the overall time complexity of the program \texttt{bellman-ford} is $\mathrm{O}(nm)$.
\end{proof}

\begin{center}\begin{center}
    \resizebox{0.82\linewidth}{!}{\centering
\begin{minipage}{3.93cm}
\centering
\begin{tikzpicture}[scale=0.7]
	\node (a) at (-1.500,1.333)  [draw,circle,thick,fill=gray!50, double, double distance=2pt] {\,};
	\node (b) at (0.000,1.333)   [draw,circle,thick,fill=gray!50] {\,};
	\node (c) at (1.500,1.333)   [draw,circle,thick,fill=gray!50] {\,};
	\node (d) at (-1.500,0.000)  [draw,circle,thick,fill=gray!50] {\,};
	\node (e) at (0.000,0.000)   [draw,circle,thick,fill=gray!50] {\,};
	\node (f) at (1.500,0.000)   [draw,circle,thick,fill=gray!50] {\,};
	\node (g) at (-1.500,-1.333) [draw,circle,thick,fill=gray!50] {\,};
	\node (h) at (0.000,-1.333)  [draw,circle,thick,fill=gray!50] {\,};
	\node (i) at (1.500,-1.333)  [draw,circle,thick,fill=gray!50] {\,};
	
	\draw (a) edge[->, thick] (b)
	      (a) edge[->, thick] (d)
	      (b) edge[->, thick] (c)
	      (b) edge[->, thick] (e)
	      (c) edge[->, thick] (f)
	      (d) edge[->, thick] (e)
	      (d) edge[->, thick] (g)
	      (e) edge[->, thick] (f)
	      (e) edge[->, thick] (h)
	      (f) edge[->, thick] (i)
	      (g) edge[->, thick] (h)
	      (h) edge[->, thick] (i);
\end{tikzpicture}
\captionof{figure}{Rooted square grid.}
\vspace{0.5cm}
\label{fig:graph-class-1b-r}
\end{minipage}
\begin{minipage}{3.93cm}
\centering
\begin{tikzpicture}[scale=0.7]
	\node (a) at (0.000,1.333)   [draw,circle,thick,fill=gray!50, double, double distance=2pt] {\,};
	\node (b) at (1.333,0.000)   [draw,circle,thick,fill=gray!50] {\,};
	\node (c) at (-1.333,0.000)  [draw,circle,thick,fill=gray!50] {\,};
	\node (d) at (2.000,-1.333)  [draw,circle,thick,fill=gray!50] {\,};
	\node (e) at (0.666,-1.333)  [draw,circle,thick,fill=gray!50] {\,};
	\node (f) at (-0.666,-1.333) [draw,circle,thick,fill=gray!50] {\,};
	\node (g) at (-2.000,-1.333) [draw,circle,thick,fill=gray!50] {\,};
	
	\draw (a) edge[->, thick] (b)
	      (a) edge[->, thick] (c)
	      (b) edge[->, thick] (d)
	      (b) edge[->, thick] (e)
	      (c) edge[->, thick] (f)
	      (c) edge[->, thick] (g);
\end{tikzpicture}
\captionof{figure}{Rooted binary tree.}
\vspace{0.5cm}
\label{fig:graph-class-1c-r}
\end{minipage}
\begin{minipage}{3.93cm}
\centering
\begin{tikzpicture}[scale=0.7]
	\node (a) at (0.000,0.000)   [draw,circle,thick,fill=gray!50, double, double distance=2pt] {\,};
	\node (b) at (0.000,1.333)   [draw,circle,thick,fill=gray!50] {\,};
	\node (c) at (0.943,0.943)   [draw,circle,thick,fill=gray!50] {\,};
	\node (d) at (1.333,0.000)   [draw,circle,thick,fill=gray!50] {\,};
	\node (e) at (0.943,-0.943)  [draw,circle,thick,fill=gray!50] {\,};
	\node (f) at (0.000,-1.333)  [draw,circle,thick,fill=gray!50] {\,};
	\node (g) at (-0.943,-0.943) [draw,circle,thick,fill=gray!50] {\,};
	\node (h) at (-1.333,0.000)  [draw,circle,thick,fill=gray!50] {\,};
	\node (i) at (-0.943,0.943)  [draw,circle,thick,fill=gray!50] {\,};
	
	\draw (a) edge[->, thick] (b)
	      (c) edge[->, thick] (a)
	      (a) edge[->, thick] (d)
	      (e) edge[->, thick] (a)
	      (a) edge[->, thick] (f)
	      (g) edge[->, thick] (a)
	      (a) edge[->, thick] (h)
	      (i) edge[->, thick] (a);
\end{tikzpicture}
\captionof{figure}{Rooted \\star graph.}
\vspace{0.5cm}
\label{fig:graph-class-2a-r}
\end{minipage}
\begin{minipage}{3.93cm}
\centering
\begin{tikzpicture}[scale=0.7]
	\node (a) at (0.0000,1.3333)   [draw,circle,thick,fill=gray!50, double, double distance=2pt] {\,};
	\node (b) at (1.1545,0.6666)   [draw,circle,thick,fill=gray!50] {\,};
	\node (c) at (1.1545,-0.6666)  [draw,circle,thick,fill=gray!50] {\,};
	\node (d) at (0.0000,-1.3334)  [draw,circle,thick,fill=gray!50] {\,};
	\node (e) at (-1.1545,-0.6666) [draw,circle,thick,fill=gray!50] {\,};
	\node (f) at (-1.1545,0.6666)  [draw,circle,thick,fill=gray!50] {\,};
	
	\draw (a) edge[->, thick] (b)
	      (b) edge[->, thick] (c)
	      (c) edge[->, thick] (d)
	      (d) edge[->, thick] (e)
	      (e) edge[->, thick] (f)
	      (f) edge[->, thick] (a);
\end{tikzpicture}
\captionof{figure}{Rooted \\cycle graph.}
\vspace{0.5cm}
\label{fig:graph-class-2b-r}
\end{minipage}}
    \resizebox{1\linewidth}{!}{\begin{minipage}{4.88cm}
\centering
\begin{tikzpicture}[scale=0.7]
	\node (a) at (-1.500,1.333)  [draw,circle,thick,fill=gray!50, double, double distance=2pt] {\,};
	\node (b) at (1.500,1.333)   [draw,circle,thick,fill=gray!50] {\,};
	\node (c) at (-1.500,-1.333) [draw,circle,thick,fill=gray!50] {\,};
	\node (d) at (1.500,-1.333)  [draw,circle,thick,fill=gray!50] {\,};
	
	\draw 
	      (a) edge[->, thick, bend right=10=10] (b)
	      (a) edge[->, thick, bend right=10] (c)
	      (a) edge[->, thick, bend right=10] (d)
	      (b) edge[->, thick, bend right=10] (a)
	      (b) edge[->, thick, bend right=10] (c)
	      (b) edge[->, thick, bend right=10] (d)
	      (c) edge[->, thick, bend right=10] (a)
	      (c) edge[->, thick, bend right=10] (b)
	      (c) edge[->, thick, bend right=10] (d)
            (d) edge[->, thick, bend right=10] (a)
	      (d) edge[->, thick, bend right=10] (b)
	      (d) edge[->, thick, bend right=10] (c);
\end{tikzpicture}
\captionof{figure}{Rooted complete graph.}
\label{fig:graph-class-3a-r}
\end{minipage}
\hspace{2cm}
\begin{minipage}{4.88cm}
\centering
\begin{tikzpicture}[scale=0.7]
	\node (a) at (-2.5,0)  [draw,circle,thick,fill=gray!50, double, double distance=2pt] {\,};
	\node (b) at (-0.8333,0)   [draw,circle,thick,fill=gray!50] {\,};
	\node (c) at (0.8333,0) [draw,circle,thick,fill=gray!50] {\,};
	\node (d) at (2.500,0)  [draw,circle,thick,fill=gray!50] {\,};
        \node (X) at (0,1.333){\,};
        \node (y) at (0,-1.333){\,};
	
	\draw (a) edge[->, thick] (b)
	      (b) edge[->, thick] (c)
	      (c) edge[->, thick] (d);
\end{tikzpicture}
\captionof{figure}{Rooted linked list.}
\label{fig:graph-class-3b-r}
\end{minipage}
\hspace{2cm}
\begin{minipage}{4.88cm}
    \centering
    \begin{tikzpicture}[scale=0.7]
        \node (b) at (-0.8333,0)   [draw,circle,thick,fill=gray!50, double, double distance=2pt] {\,};
        \node (c) at (0.8333,0) [draw,circle,thick,fill=gray!50] {\,};
        \node (d) at (2.500,0)  [draw,circle,thick,fill=gray!50] {\,};
            \node (X) at (0,1.333){\,};
            \node (y) at (0,-1.333){\,};

    \end{tikzpicture}
    \captionof{figure}{Rooted discrete graph.}
    \label{fig:graph-class-6b-r}
\end{minipage}}
    \end{center}
\end{center}
\vspace{0.5em}

We supplement the proof with empirical measurements on the graph classes of Figures \ref{fig:graph-class-1b-r}, \ref{fig:graph-class-1c-r}, \ref{fig:graph-class-2a-r}, \ref{fig:graph-class-2b-r}, \ref{fig:graph-class-3a-r}, \ref{fig:graph-class-3b-r} and \ref{fig:graph-class-6b-r}. In Figures \ref{fig:bench-bf-random} and \ref{fig:bench-bf}, all input graph edges are labelled with an integer chosen uniformly at random in the interval $[-100, 100]$; the allocation of weights was done such that cycle graphs do not contain negative-weight cycles. And finally, Figure \ref{fig:bench-bf-neg} shows the runtimes on cycle graphs whose weights alternate between $-2$ and $1$ such that the total weight is negative. Hence, the program fails on these inputs.

The low runtime of complete graphs, compared with the other graph classes, is due to the fact that the number of edges is quadratic in the number of nodes, and hence the runtime bound $nm$ grows at a slower rate (note that the scales of the plots are different). 

\begin{figure}[!ht]
    \centering
    \begin{tikzpicture}
    \begin{axis}[
      xlabel=size (number of nodes and edges),
      ylabel=runtime (ms), ylabel style={above=0.2mm},
      width=9.2cm,height=7.2cm,
      legend style={at={(1.65,1)}},
      ymajorgrids=true,
      grid style=dashed]
      \addplot[color=plot2, mark=square*] table [y=time, x=n]{Figures/Benchmarks/BF-random/bf-list.dat};
      \addlegendentry{List graphs}
      \addplot[color=plot1, mark=square*] table [y=time, x=n]{Figures/Benchmarks/BF-random/bf-cycle.dat};
      \addlegendentry{Cycle graphs}
      \addplot[color=plot5, mark=square*] table [y=time, x=n]{Figures/Benchmarks/BF-random/bf-tree.dat};
      \addlegendentry{Binary trees}
      \addplot[color=plot7, mark=square*] table [y=time, x=n]{Figures/Benchmarks/BF-random/bf-grid.dat};
      \addlegendentry{Square grids}
        \addplot[color=plot4, mark=square*] table [y=time, x=n]{Figures/Benchmarks/BF-random/bf-star.dat};
      \addlegendentry{{Star graphs}}
    \end{axis}  
    \end{tikzpicture}
    \caption{Measured performance of the program \texttt{bellman-ford} on rooted graph classes with random weights.}
    \label{fig:bench-bf-random}
\end{figure}
In Figure \ref{fig:bench-bf}, we include the runtimes for discrete graphs to stress that even for these simple graphs, the time complexity is not constant but linear. That is because the loop \texttt{count!} does require time linear in the number of nodes.

\begin{figure}[!ht]
    \centering
    \begin{tikzpicture}
    \begin{axis}[
      xlabel=size (number of nodes and edges),
      ylabel=runtime (ms), ylabel style={above=0.2mm},
      width=9.2cm,height=7.2cm,
      legend style={at={(1.65,1)}},
      ymajorgrids=true,
      grid style=dashed]
      \addplot[color=red, mark=square*] table [y=time, x=n]{Figures/Benchmarks/BF/bf-discrete.dat};
      \addlegendentry{Discrete graphs}
      \addplot[color=plot6, mark=square*] table [y=time, x=n]{Figures/Benchmarks/BF/bf-complete.dat};
      \addlegendentry{{Complete graphs}}
    \end{axis}  
    \end{tikzpicture}
    \caption{Measured performance of the program \texttt{bellman-ford} on rooted discrete graphs, and rooted complete graphs with random weights.}
    \label{fig:bench-bf}
\end{figure}
\begin{figure}[!ht]
    \centering
    \begin{tikzpicture}
    \begin{axis}[
      xlabel=size (number of nodes and edges),
      ylabel=runtime (ms), ylabel style={above=0.2mm},
      width=9.2cm,height=7.2cm,
      legend style={at={(1.65,1)}},
      ymajorgrids=true,
      grid style=dashed]
      \addplot[color=plot2, mark=square*] table [y=time, x=n]{Figures/Benchmarks/BF-alternating/bf-alter-cycle.dat};
    \end{axis}  
    \end{tikzpicture}
    \caption{Measured performance of the program \texttt{bellman-ford} on rooted cycle graphs of negative total weight.}
    \label{fig:bench-bf-neg}
\end{figure}

\section{Conclusion}
\label{s:conclusion}
We have demonstrated by case studies how to implement three fundamental graph algorithms in \gp{} such that their imperative time complexities are matched, even if input graphs have an unbounded node degree or are possibly disconnected. Addressing the issues of unbounded degree and disconnectedness has been an open problem since the publication of the first paper on rooted graph transformation \cite{Bak-Plump12a}. Up to now, only certain graph reduction programs that destroy their input graphs could be constructed to run in linear time on graph classes that have an unbounded node degree or contain disconnected graphs \cite{campbell2022fast}. In Subsections~\ref{ss:first-enhancement} and \ref{ss:second-enhancement}, we show how the graph data structure generated by the new compiler enables linear runtimes on arbitrary GP\,2 host graphs.

Our case studies include an implementation of the Bellman-Ford single-source shortest-path algorithm that matches the $\mathrm{O}(nm)$ time complexity of imperative implementations. This is, to the best of our knowledge, the first time that a rule-based shortest-path algorithm achieves the complexity of a corresponding algorithm in an imperative programming language.

Our approach involves both enhancing the graph data structure generated by the \gp{} compiler and developing a programming technique that leverages the new graph representation. Previously, the graph data structure in the C program generated by the compiler stored in each host-graph node two linked lists of incoming and outgoing edges. In contrast, the new two-dimensional array of linked lists of edges allows to find an edge with a given mark and orientation in constant time. Also, nodes are now stored in multiple linked lists, one for each mark. Thus, a node of a given mark can be found in constant time.

Comparing our proofs of correctness and complexity in the case studies with corresponding proofs for imperative programs in standard textbooks such as \cite{Cormen-Leiserson-Rivest-Stein22a}, the latter may appear more concise. However, one should note that those proofs are at the pseudo-code level, whereas ours address all details at the level of concrete syntax, including the specifics of the compiler implementation.

We speculate that many more graph algorithms can be implemented as \gp{} programs running under time bounds achievable with conventional programming languages. To make progress in developing such algorithms, we intend to implement a \gp{} library of advanced data structures, such as priority queues, Fibonacci heaps, or AVL trees, supporting programmers in constructing \gp{} versions of efficient imperative graph algorithms.

\appendix

\section{Abstract Syntax of Labels and Conditions}
\label{s:appendix}
This appendix gives grammars in Extended Backus-Naur Form defining the abstract syntax of labels and application conditions. There are some context constraints for these definitions which can be found in \cite{Bak15a}.

The class of \emph{host graphs} consists of all totally labelled graphs over the set HostLabel of host graph labels as defined by the grammar in Figure \ref{fig:host-label-grammar}.

\begin{figure}[h!]
\centering
\renewcommand{\arraystretch}{1.2}
\begin{tabular}{lcl}
HostLabel & ::= & HostList [HostMark] \\
HostList & ::= & \ttt{empty} $\mid$ HostAtom \{`:' HostAtom\} \\
HostAtom & ::= & HostInteger $\mid$ HostString \\ 
HostInteger & ::= & [`-'] Digit \{Digit\} \\ 
HostString & ::= & `{``}\,'\{HostChar\}`\,{''}' \\ 
HostChar & ::= & `{``}\,'Character`\,{''}' \\ 
HostMark & ::= & \ttt{red} $\mid$ \ttt{green} $\mid$ \ttt{blue} $\mid$ \ttt{grey} $\mid$ \ttt{dashed} 
\end{tabular}
\caption{Abstract syntax of host graph labels. \label{fig:host-label-grammar}
}
\end{figure}

Figure \ref{fig:rule-label-grammar} defines the labels that may occur in the graphs contained in rules. (In the case studies of this paper, however, we only need labels that are variables of type \ttt{list} or \ttt{int}.) We obtain the class of \emph{rule graphs} by replacing the label set $\L$ of Definition \ref{def:graph} with the set Label of rule graph labels. 

\begin{figure}[h!] 
\centering
\renewcommand{\arraystretch}{1.2}
\begin{tabular}{lcl}
Label & ::= & List [Mark] \\
List & ::= & LVar $\mid$ \ttt{empty} $\mid$ Atom $\mid$ List `:' List \\
Atom & ::= & AVar $\mid$ Integer $\mid$ String \\
Integer & ::= & IVar $\mid$ [`-'] Digit \{Digit\} $\mid$ `('Integer`)' $\mid$ \\
 & & Integer (`+' $\mid$ `-' $\mid$ `$\ast$' $\mid$ `/') Integer $\mid$ \\ 
 & & (\ttt{indeg} $\mid$ \ttt{outdeg}) `('Node`)' $\mid$ \\
 & & \ttt{length} `('(LVar $\mid$ AVar $\mid$ SVar)`)' \\
String & ::= & SVar $\mid$ Char $\mid$ `{``}\,'\{Character\}`\,{''}' $\mid$ \\
 & & String `\texttt{.}' String \\
Char & ::= & CVar $\mid$ `{``}\,'Character`\,{''}' \\
Mark & ::= & \ttt{red} $\mid$ \ttt{green} $\mid$ \ttt{blue} $\mid$ \ttt{grey} $\mid$ \ttt{dashed} $\mid$ \ttt{any} 
\end{tabular}
\caption{Abstract syntax of rule graph labels. \label{fig:rule-label-grammar}
}
\end{figure}

The concatenation of two lists $x$ and $y$ is written $x{:}y$ (not to be confused with Haskell's colon operator which adds an element to the beginning of a list). The binary arithmetic operators `\texttt{+}', `\texttt{-}', `$\ast$' and `\texttt{/}' expect integer expressions as arguments while `\texttt{.}' is string concatenation. The categories LVar, AVar, IVar, SVar and CVar consist of variables of type List, Atom, Integer, String and Char, respectively. Variables can be used as expressions of supertypes. For example, integer variables can be used as lists of length one. The category Node consists of identifiers for nodes in the left-hand graphs of rules. The operators \texttt{indeg} and \texttt{outdeg} denote the number of ingoing respectively outgoing edges of a node. The \texttt{length} operator returns the length of a list or string represented by a variable. The category Character consists of the printable ASCII characters except ‘\,{"}’. Character strings are enclosed in double quotes. 

The grammar in Figure \ref{fig:condition-grammar} defines the abstract syntax of rule application conditions. Conditions are Boolean combinations of subtype assertions, applications of the \ttt{edge} predicate to left-hand node identifiers, or relational comparisons of expressions (where \ttt{=} and \ttt{!=} can be used for arbitrary expressions). A condition of the form $\mtt{edge}(v_1,v_2,l)$ requires the existence of an edge from node $v_1$ to node $v_2$ with label $l$.
   
\begin{figure}[h!]
\centering
\renewcommand{\arraystretch}{1.2}
\begin{tabular}{lcl}
Condition & ::= & (\ttt{int} $\mid$ \ttt{char} $\mid$ \ttt{string} $\mid$ \ttt{atom}) `('(LVar $\mid$ AVar $\mid$ SVar)`)' \\
&& $\mid$ List (`\ttt{=}' $\mid$ `\ttt{!=}') List \\
&& $\mid$ Integer (`\texttt{>}' $\mid$ `\texttt{>=}' $\mid$ `\texttt{<}' $\mid$ `\texttt{<=}') Integer \\
&& $\mid$ \ttt{edge} `(' Node `,' Node  [`,' Label] `)' \\
&& $\mid$ \ttt{not} Condition \\
&& $\mid$ Condition (\texttt{and} $\mid$ \texttt{or}) Condition \\
&& $\mid$  `(' Condition `)'
\end{tabular}
\caption{Abstract syntax of application conditions. \label{fig:condition-grammar}}
\end{figure}

\newpage 

\section*{Acknowledgements}
  \noindent The authors thank Brian Courtehoute and Robert S\"oldner for discussions on the topics of this paper, and the anonymous reviewers for their helpful comments.

\bibliographystyle{alphaurl}
\bibliography{main}

\end{document}